\documentclass{article}

\usepackage{arxiv}

\usepackage[utf8]{inputenc} 
\usepackage[T1]{fontenc}    
\usepackage{hyperref}       
\usepackage{url}            
\usepackage{booktabs}       
\usepackage{amsfonts}       
\usepackage{nicefrac}       
\usepackage{microtype}      
\usepackage{lipsum}
\usepackage{amsthm}

\input{defi}
\newtheorem{lemma}{Lemma}

\title{Quantifying identifiability to choose and audit \eps in differentially private deep learning}

\author{
  Daniel Bernau, Hannah Keller\thanks{Authors contributed equally.} \\
  SAP\\
  Karlsruhe, Germany \\
  \texttt{firstname.lastname@sap.com} \\
   \And
  G\"{u}nther Eibl\footnotemark[1] \\
  Salzburg University of Applied Sciences \\
  Salzburg, Austria \\
\texttt{guenther.eibl@en-trust.at} \\
   \And
  Philip W.~Grassal\footnotemark[1] \thanks{Work was done during an internship at SAP.} \\
  University of Heidelberg \\
  Heidelberg, Germany \\
  \texttt{philip-william.grassal@iwr.uni-heidelberg.de}\\
   \And
  Florian Kerschbaum \\
  University of Waterloo \\
  Waterloo, Canada \\
\texttt{florian.kerschbaum@uwaterloo.ca} \\
}

\begin{document}
\maketitle

\begin{abstract}
{Differential privacy allows bounding the influence that training data records have on a machine learning model. To use differential privacy in machine learning, data scientists must choose privacy parameters~\epsdlt. 
Choosing meaningful privacy parameters is key, since models trained with weak privacy parameters might result in excessive privacy leakage, while strong privacy parameters might overly degrade model utility. However, privacy parameter values are difficult to choose for two main reasons. 
First, the theoretical upper bound on privacy loss \epsdlt might be loose, depending on the chosen sensitivity and data distribution of practical datasets.
Second, legal requirements and societal norms for anonymization often refer to individual identifiability, to which \epsdlt are only indirectly related. 

We transform \epsdlt to a bound on the Bayesian posterior belief of the adversary assumed by differential privacy concerning the presence of any record in the training dataset. The bound holds for multidimensional queries under composition, and we show that it can be tight in practice. Furthermore, we derive an identifiability bound, which relates the adversary assumed in differential privacy to previous work on membership inference adversaries. We formulate an implementation of this differential privacy adversary that allows data scientists to audit model training and compute empirical identifiability scores and empirical \epsdlt. }
\end{abstract}

\section{Introduction}
\label{sec:introduction}
The application of differential privacy (DP) for machine learning has received considerable attention by the privacy research community, leading to key contributions such as the tight estimation of privacy loss under composition~\cite{Mironov2017,KOV15,KOV17} and differentially private stochastic gradient descent~\cite{abadi2016,SCS13,BST14,shokri2015} (DPSGD) for training neural networks. Still, data scientists must choose privacy parameters \epsdlt to train a machine learning (ML) model using DPSGD. If the privacy parameters are stronger than necessary, model utility is sacrificed, as these parameters only formulate a theoretic upper bound that might not be reached when training an ML model with differentially private stochastic gradient descent on real-world data. If privacy parameters are too small, the trained model might be prone to reidentification attacks.

Several privacy regulations~\cite{Parliament2016,hhs} consider individual identifiability to gauge anonymization strength. Therefore, relatable scores that quantify reidentification risk to individuals can strongly affect the widespread implementation of anonymization techniques~\cite{Nissenbaum}. In consequence, if DP shall be used to comply with privacy regulations and find widespread adoption~\cite{Nissim2018, anonymity_opinion}, quantifying the resulting identifiability from privacy parameters \epsdlt is required~\cite{Nissim2018,Clifton2013}.

Multiple approaches for choosing privacy parameters have been introduced, yet they do not reflect identifiability~\cite{abowd2019,Hsu2014}, part from the original DP definition~\cite{YGF+18,Lee2012,RRL+18,Bernau2019}, or lack applicability to common DP mechanisms for ML~\cite{Lee2011}. Especially in ML, practical membership inference (MI) attacks have been used to measure identifiability~\cite{Bernau2019, RRL+18, Hayes2019, JE19, JWK+20, chen2020, shokri2017, YGF+18}. However, MI adversaries are not assumed to have auxiliary information about the members of datasets that they aim to differentiate, which DP adversaries are assumed to possess. 
MI attacks thus offer intuition about the outcome of practical attacks; nonetheless, bounds on MI attacks in terms of \eps are not tight~\cite{JE19}, and consequently MI can only represent an empirical lower bound on identifiability.

Rather than analyzing the MI adversary, we consider a DP adversary with arbitrary auxiliary knowledge and derive maximum \textit{Bayesian posterior belief} $\rb$ as an identifiability bound related to \epsdlt, which bounds the adversary's certainty in identifying a member of the training data. Furthermore, we define the complementary score \textit{expected membership advantage} $\ra$, which is related to the probability of success in a Bernoulli trial over the posterior beliefs. $\ra$ depends on the entire distribution of observed posterior beliefs, not solely the worst case posterior belief, and allows direct comparison with the membership advantage bound of Yeom et al.~\cite{YGF+18} for the MI adversary. We will show that the DP adversary achieves greater membership advantage than an MI adversary, implying that while both adversaries can be used to evaluate the protection of DP in machine learning, our implementable instance of the DP adversary comes closer to DP bounds.

A subsequent question is whether our identifiability bounds are tight in practice, since the factual guarantee \epsdlt~depends on the difference between possible input datasets~\cite{Nissim2007}. In differentially private stochastic gradient descent, noise is scaled to global sensitivity, the maximum change that any single record in the training dataset is assumed to cause on the gradient during any training step. However, since all training data records are likely to be within the same domain (e.g., pictures of cars vs.~pictures of nature scenes), global sensitivity might far exceed the difference between gradients over all training steps. 
We propose scaling the sensitivity to the difference between the gradients of a fixed dataset and any neighboring dataset and show for three reference datasets that we can indeed achieve tight bounds.
Our main contributions are:
	\begin{itemize}
		\item Identifiability bounds for the posterior belief and expected membership advantage that are mathematical transformations of privacy parameters \epsdlt and used in conjunction with RDP composition.
		\item The practical implementation of an adversary that meets all assumptions on worst-case adversaries against DP and allows us to audit DPSGD model instances w.r.t.~to the empirical privacy loss besides enabling comparison with membership inference adversaries.
		\item A heuristic for scaling sensitivity in differentially private stochastic gradient descent. This heuristic leads to tight bounds on identifiability.
	\end{itemize}
	
This paper is structured as follows. Preliminaries are presented in Section~\ref{sec:preliminaries}. We formulate identifiability scores and provide upper bounds on them in Sections~\ref{sec:interpreting-dp} and~\ref{sec:up_bound}. Section \ref{sec:tight} specifies the application of these scores for a deep learning scenario, and we evaluate the scores for three deep learning reference datasets in Section~\ref{sec:eval}. Section~\ref{sec:disc} discusses the practical relevance of our findings. We present related work and conclusions in Sections~\ref{sec:rel} and~\ref{sec:conc}.
\section{Preliminaries}
\label{sec:preliminaries}
\subsection{Differential Privacy}
\label{sec:dp}

If the evaluation of a function $f: \cali{U} \rightarrow \cali{R}$ on a dataset \cali{D} from domain $\cali{U}$ yields a result $r$, $r$ inevitably leaks information about the entries $\datapoint \in \cali{D}$ (cf.~impossibility of Dalenius' desideratum~\cite{dwork2006a}). DP~\cite{dwork2006a} offers an anonymization guarantee for statistical query functions $f(\cdot)$, perturbing $r$ such that the result could have been produced from dataset \cali{D} or some \textit{neighboring} dataset \cali{D'}. 
A neighboring dataset \cali{D'} either differs from \cali{D} in the presence of one additional data point (unbounded DP) or in the value of one data point when a data point from \cali{D} is replaced by another data point (bounded DP). 
In the context of this work, we will consider w.l.o.g.~ unbounded DP where $\cali{D}$ contains one data point $\datapoint$ more than $\cali{D'}$ and $\cali{D}\setminus\cali{D'}=\datapoint$. To achieve differential privacy, noise is added to the result of $f(\cdot)$ by \textit{mechanisms} \cali{M} according to Definition~\ref{def:DP}. The impact of a single member $\datapoint \in \cali{D}$ on $f(\cdot)$ is bounded. If this impact is low compared to the noise specified by DP, plausible deniability is provided to this member of $\cali{D}$, even if $\cali{D}$ and the members' properties $\datapoint$ (and thus also $\cali{D}'$) are known. For example, a single individual participating in a private analysis based on a census income dataset such as Adult~\cite{Koh96} could therefore plausibly deny census participation and values of personal attributes. 
DP provides a strong guarantee, since it protects against a strong adversary with knowledge of up to all points in a dataset except one. As Definition~\ref{def:DP} is an inequality, the privacy parameter $\eps$ can be interpreted as an upper bound on privacy loss.

\begin{definition}[\epsdlt-Differential Privacy~\cite{dwork2006b}]
	\label{def:DP}
A mechanism \cali{M} preserves \epsdlt-differential privacy if for all independently sampled $\cali{D},\cali{D'}\subseteq \cali{U}$, where $\cali{U}$ is a finite set, with $\cali{D}$ and $\cali{D'}$ differing in at most one element, and all possible mechanism outputs $\cali{S}$
	 \begin{equation*}
	 \Pr(\cali{M}(\cali{D}) \in \cali{S}) \leq e^{\eps} \cdot \Pr(\cali{M}(\cali{D'}) \in \cali{S}) + \dlt
	 \end{equation*}
\end{definition}

The Gaussian mechanism is the predominant DP mechanism in ML for perturbing the outcome of stochastic gradient descent and adds noise independently sampled from a Gaussian distribution centered at zero. Prior work~\cite{dwork2014} has analyzed the tails of the normal distributions and found that bounding the standard deviation as follows fulfills \epsdlt-DP:
\begin{equation} 
\label{eq:sigmaEpsDlt}
\sgm > \sens\sqrt{2\ln(1.25/\dlt)}/\eps
\end{equation}

Rearranged to solve for $\epsilon$, this is:
\begin{equation}
    \label{eq:eps_gauss}
    \eps > \sens\sqrt{2\ln(1.25/\dlt)}/\sgm
\end{equation}

$\sgm$ depends not only on the DP guarantee, but also on a scaling factor $\sens$. $\sens$ is commonly referred to as the sensitivity of a query function $f(\cdot)$ and comes in two forms: global sensitivity $GS_f$ and local sensitivity $LS_f$. DP holds if mechanisms are scaled to $GS_f$ of Definition~\ref{def:global_sensitivity}, i.e.,~the maximum contribution of a record in the dataset to the outcome of $f(\cdot)$.

\begin{definition}[Global Sensitivity]
\label{def:global_sensitivity}
Let $\cali{D}$ and $\cali{D'}$ be neighboring. For a given finite set $\cali{U}$ and function $f$ the global sensitivity $GS_f$ with respect to a distance function is 
\begin{equation*}
GS_f = \max\limits_{\cali{D}, \cali{D'}}||f(\cali{D})-f(\cali{D}')||
\end{equation*}
\end{definition}

For the Gaussian mechanism, we use the global $\ell_2$-sensitivity $GS_{f_2}$. Local sensitivity is specified in Definition~\ref{def:local_sensitivity}~\cite{Nissim2007} and fixes dataset $\cali{D}$. Note that the absolute $GS_f$ as of Definition~\ref{def:global_sensitivity} can also be defined relative to local sensitivity, as $GS_f= \max\limits_\cali{D} LS_f(\cali{D})$. The impact of $LS_f$ is that, compared to using $GS_f$, less noise is added when $\eps$ is held constant, and \eps is decreased when the noise distribution is held constant.

\begin{definition}[Local Sensitivity]
\label{def:local_sensitivity}
Let $\cali{D}$ and $\cali{D'}$ be neighboring. For a given finite set $\cali{U}$, independently sampled dataset $ \cali{D}\subseteq$ $\cali{U}$, and function $f$, the local sensitivity $LS_f(\cali{D})$ with respect to a distance function is 
\begin{equation*}
LS_f(\cali{D}) = \max\limits_{\cali{D}'}||f(\cali{D})-f(\cali{D}')||
\end{equation*}
\end{definition}

In differentially private stochastic gradient descent, perturbed outputs are released repeatedly in an iterative process. \cali{M} is represented by a differentially private version of an ML optimizer such as Adam or SGD. The most basic form of accounting multiple data releases is sequential composition, which states that for a sequence of $k$ mechanism executions each providing ($\eps_i$, $\dlt_i$)-DP, the total privacy guarantee composes to ($\sum_{i}\eps_{i}$, $\sum_{i}\dlt_{i}$)-DP; however, sequential composition adds more noise than necessary~\cite{abadi2016,Mironov2017}. 

A tighter analysis of composition is provided by Mironov~\cite{Mironov2017}. $(\alpha,\eps_{RDP})-$R{\'e}nyi differential privacy (RDP), with $\alpha >1$ quantifies the difference in distributions $\cali{M}(\cali{D}), \cali{M}(\cali{D'})$ by their R{\'e}nyi divergence~\cite{vanErven2010}. For a sequence of $k$ mechanism executions each providing ($\alpha$, $\eps_{RDP, i}$)-RDP, the privacy guarantee composes to ($\alpha$, $\sum_{i}\eps_{RDP, i}$)-RDP. The ($\alpha$, $\eps_{RDP}$)-RDP guarantee converts to $(\eps_{RDP}-\frac{\ln\dlt}{\alpha-1},\dlt)$-DP. The Gaussian mechanism is calibrated to RDP by:

\begin{equation}\label{eq:gaussrdp}
    \eps_{RDP} = \alpha \cdot \sens^2/2\sgm^2
\end{equation}

\subsection{Differentially Private Machine Learning} \label{subsec:problemStatement}
In machine learning a neural network (NN) is commonly provided a training dataset $\cali{D}$ where each of the data points $(x,y)\in\cali{D}$ consists of the features $x$ and the label $y$. The goal is to learn a prediction function using an optimizer. A test set is used to evaluate generalization and utility of the trained model. This paper focuses on applying DP to stochastic gradient descent optimizers that output a gradient vector, which corresponds to the output of function $f$ in DP.
A variety of differentially private stochastic gradient descent (DPSGD) optimizers are available for deep learning, all of which depend on the privacy parameters $\epsdlt$ and the clipping norm $\clip$~\cite{abadi2016,MAE+19}. DPSGD updates weights $\theta_i$ of the NN per training step $i \in k$ with $\theta_i \leftarrow \theta_{i-1}-\eta\cdot \dpgrad_i$, where $\eta>0$ is the learning rate. Differential privacy is achieved by perturbing the gradient $\dpgrad_i = \cali{M}_i(\cali{D})$ with Gaussian noise. To limit the sensitivity $\sens$, the length of each per-example gradient is limited to the clipping norm \clip before perturbation, and the Gaussian perturbation is proportional to \clip.

\subsection{Membership Inference}
\label{subsec:mi}

Membership inference (MI) is a threat model for quantifying how accurately an adversary can identify members of the training data in ML. Yeom et al.~\cite{YGF+18} formalize MI in the following experiment:

\begin{experiment} (Membership Inference $\emi$)
	Let $\Ami$ be an adversary, $\cali{M}$ be a differentially private learning algorithm, $n$ be a positive integer, and $\var{Dist}$ be a distribution over data points $(x,y)$. Sample $\cali{D}\sim\var{Dist}^n$ and let $\vec{r}=\cali{M}(\cali{D})$. The membership experiment proceeds as follows:
	\begin{enumerate}
	    \item Sample $z_{\cali{D}}$ uniformly from $\cali{D}$ and $z_{\var{Dist}}$ from $\var{Dist}$
		\item Choose $b\leftarrow\{0,1\}$ uniformly at random
		\item Let
		\[
		z= \begin{cases}
		z_{\cali{D}} \quad & \mathrm{if} \ $b=1$\\
		z_{\var{Dist}} & \mathrm{if} \ $b=0$
		\end{cases}
		\]
		\item $\Ami$ outputs $b'= \Ami(\vec{r}, z, \var{Dist}, n, \cali{M}) \in \{0,1\}$. If $b'=b$, $\Ami$ succeeds and the output of the experiment is 1. It is 0 otherwise.
	\end{enumerate}
	\label{exp:MI}
\end{experiment}

\subsection{Differential Identifiability and the relation to the DP adversary}
\label{subsec:adversary}

Lee et al.~\cite{Lee2011, Lee2012} introduce differential identifiability (DI) as a strong inference threat model. DI assumes that the adversary calculates the likelihood of all possible input datasets, so-called \textit{possible worlds} in a set $\Psi$, given a mechanism output $r$. Li et al.~\cite{Li13} show that the DI threat model maps to the worst case against which bounded DP protects when $|\Psi| = 2$, since DP considers two neighboring datasets $\cali{D}$, $\cali{D'}$ by definition. 
The DI experiment $\edi$ is similar to $\emi$, since the adversary must decide whether the dataset contains the member that differs between the known $\cali{D'}$ and $\cali{D}$, or not. For comparison we reformulate DI as a cryptographic experiment:
\begin{experiment}
	(Differential Identifiability $\edi$) Let $\Astrong$ be an adversary, $\cali{M}$ be a differentially private learning algorithm, $\cali{D}$ and $\cali{D'}$ be neighboring datasets drawn mutually independently from distribution $\var{Dist}$, using either bounded or unbounded definitions. The differential identifiability experiment $\edi$ proceeds as follows:
	\begin{enumerate}
	    \item Set $\vec{r}_{{D}}:=\cali{M}(\cali{D})$ and $\vec{r}_{\cali{D'}}:=\cali{M}(\cali{D'})$
		\item Choose $b\leftarrow\{0,1\}$ uniformly at random
		\item Let 
\[
\vec{r}= \begin{cases}
\vec{r}_{\cali{D}}, \quad & \mathrm{if} \ b=1 \\
\vec{r}_{\cali{D'}}, & \mathrm{if} \ b=0
\end{cases}
\]
\item $\Astrong$ outputs $b'= \Astrong(\vec{r}, \cali{D},\cali{D'},
\cali{M}, \var{Dist}) \in \{0,1\}$. If $b'=b$, $\Astrong$ succeeds and the output of the experiment is 1. It is 0 otherwise.
	\end{enumerate}
	\label{exp:ident}
\end{experiment}

Since Experiment~\ref{exp:ident} precisely defines an adversary with access to arbitrary background knowledge of up to all but one record in $\cali{D}$ and $\cali{D'}$, $\Astrong$ is an implementable instance of the DP adversary~\cite{Dwork2016}. Compared to the MI adversary, the DI adversary is stronger, since $\Astrong$ knows the alternative dataset $\cali{D'}$ instead of only the distribution $\var{Dist}$ from which $\cali{D'}$ was chosen. The experiment defined above is general and applies to deep learning using gradient descent as follows: the knowledge of the mechanism $\cali{M}$ implies knowledge about the architecture of the NN and the learning parameters $\eta, \clip$, as well as number of iterations $k$. The experiment is formulated s.t.~it could be applied for a single iteration, and the output $\vec{r}$ of the mechanism is the perturbed gradient $\dpgrad_i$ from iteration $i$ of the NN training. However, after the entire learning process, consisting of $k$ rounds, $\Astrong$ has more information $R_k=(\vec{r}_0,\vec{r}_1,\ldots,\vec{r}_k)$ and therefore a higher chance to win Experiment~\ref{exp:ident}. In this case, the same value of $b$ is chosen in every round, since the training data is kept constant over all learning steps. This is the standard case considered in our paper and motivates the need for composition theorems. According Experiment~\ref{exp:ident}, the DI adversary could know almost all of the training data from a public dataset of census data, for example, and observe the NN gradient updates at every training step. The assumption that $\Astrong$ has access to all gradients during learning may seem overly strong; however, this setting is of theoretical interest, since the bounds that we prove for the DI adversary $\Astrong$ will also hold for weaker adversaries. Furthermore, the assumptions can be fulfilled in federated learning, for example. In federated learning, multiple data owners jointly train a global model by sharing gradients for their individual subsets of training data with a central aggregator. The aggregator combines the gradients and shares the aggregated update with all data owners. If $\Astrong$ participates as a data owner, $\Astrong$ is able to observe the joint model updates.
\section{Identifiability scores for DP} 
\label{sec:interpreting-dp}

In this section we formulate two scores for identifiability of individual training records when releasing a differentially private NN. The scores are compatible with DP under multidimensional queries and composition. 
\iffullversion We prove that protection against $\Astrong$ implies protection against $\Ami$. \else We show that if $\madi$ is bounded, then this bound also holds for $\mami$. Equivalently, we show that $\Astrong$ is stronger than $\Ami$ due to additional available auxiliary information. Concretely, $\Astrong$ knows both neighboring datasets $\cali{D}$ and $\cali{D'}$ instead of only receiving one value $z$ and the size $n$ of the dataset from which the data points are drawn. This leads us to Proposition~\ref{prop:reduction}, which we formally prove by reduction in an extended version of this paper~\cite{BEGK21}.
\begin{proposition}\label{prop:reduction}
    DI implies MI: if $\Ami$ wins $\emi$, then one can construct $\Astrong$ that wins $\edi$.
\end{proposition}
\fi{}
We define \textit{posterior belief} $\beta$, which quantifies identifiability for iterative mechanisms in Section~\ref{sec:post_bel}. Second, we define \textit{membership advantage} $\madi$ for $\Astrong$ in Section~\ref{sec:adv}, which is a complementary identifiability score offering a scaled quantification of the adversary's probability of success. 
\iffullversion
\subsection{Relation of Membership Inference and Differential Identifiability}
\label{sec:adversary}

We first show that if $\madi$ is bounded, then this bound also holds for $\mami$. Equivalently, we show that $\Astrong$ is stronger than $\Ami$ due to additional available auxiliary information. Concretely, $\Astrong$ knows both neighboring datasets $\cali{D}$ and $\cali{D'}$ instead of only receiving one value $z$ and the size $n$ of the dataset from which the data points are drawn.

\begin{proposition}\label{prop:reduction}
    DI implies MI: if $\Ami$ wins $\emi$, then one can construct $\Astrong$ that wins $\edi$.
\end{proposition}

\begin{proof}
    We prove the proposition by contradiction: assume that the mechanism $\cali{M}$ successfully protects against $\Astrong$, but that there exists an adversary $\Ami$ that wins $\emi$. Again, we assume w.l.o.g.~that $\cali{D}\setminus\cali{D'}\neq\{\}$. We construct an adversary $\Astrong$ that also wins $\edi$ as follows:
    \begin{enumerate}
        \item On inputs $\cali{D}, \cali{D'}, \cali{M}, \vec{r}$, $\var{Dist}$, $\Astrong$ calculates $n=|\cali{D}|$ and let $z=\cali{D}\setminus\cali{D'}$.
        \item $\Astrong$ gives $(z, \vec{r}, n, \var{Dist})$ to $\Ami$.
        \item $\Ami$ gives $b''=\Ami(z, \vec{r}, n, \var{Dist})$ to $\Astrong$ in response.
        \item $\Astrong$ outputs $b'=b''$.
    \end{enumerate}
    By the definition of $\edi$, $\Astrong$ wins if $b'=b$, and thus succeeds in the following cases:\\
    \textbf{Case 1:} $b=1$, which means $\vec{r}=\cali{M}(\cali{D})$. Since $z\in \cali{D}$, this is exactly the case where $\Ami$ correctly outputs $b''=1$. Therefore $b'=b$.\\
   \textbf{Case 2:} $b=0$, which means $\vec{r}=\cali{M}(\cali{D'})$. Since $z\notin\cali{D'}$, this is exactly the case where $\Ami$ correctly outputs $b''=0$. Therefore $b'=b$.
    For both cases $\Astrong$ wins ($b'=b$), which contradicts the assumption that the mechanism $\cali{M}$ successfully protects against $\Astrong$. It is at least as difficult for a mechanism to protect against $\edi$ as against $\emi$, which is equivalent to the statement that if $\Ami$ wins $\emi$, then $\Astrong$ wins $\edi$ as well. 
\end{proof}
\else\fi{}

\subsection{Posterior Belief in Identifying the Training Dataset}
\label{sec:post_bel}

To quantify individual identifiability from privacy parameters \epsdlt, we use the Bayesian posterior belief. After having observed gradients $R_k$, the adversary $\Astrong$ can update the probabilities for both the training dataset $\cali{D}$ and the alternate dataset $\cali{D}'$, that differs from $\cali{D}$ in an individual record $\datapoint=\cali{D}\setminus \cali{D}'$. The posterior belief quantifies the certainty with which $\Astrong$ is able to identify the training dataset used by a NN and consequently the presence of the individual record $\datapoint$. This belief is formulated as a conditional probability depending on observations $R_k$ during training. 
For a census dataset such as Adult, the posterior belief measures the probability that a particular individual $\datapoint$ participated in the census after observing training using data $\cali{D}$. Since this belief has an upper bound for each possible member $x$ of the dataset, no member of $\cali{D}$ can be identified. 
Posterior belief therefore relates theoretical DP privacy guarantees to privacy regulations and societal norms through its identifiability formulation, since the noise, and therefore the posterior belief, depends on \epsdlt. 
\begin{definition}[Posterior Belief]

Consider the setting of Experiment~\ref{exp:ident} and denote $R_k=(\vec{r}_0,\vec{r}_1,\ldots,\vec{r}_k)$ as the result matrix, comprising $k$ multidimensional mechanism results. The posterior belief in the correct dataset $\cali{D}$ is defined as the probability conditioned on all the information observed during the adaptive computations
\[
\beta_k := \Pr(\cali{D}\vert R_k) = \frac{\Pr(\cali{D},R_k)}{\Pr(\cali{D},R_k)+\Pr(\cali{D'},R_k)}
\]
where the probability $\Pr(\cali{D}|R_k)$ is over the random iterative choices of the mechanisms up to step $k$.
\label{def:adapt-belief}
\end{definition}

Each $\beta_k$ can be computed from the previous $\beta_{k-1}$. The final belief can be computed using Lemma~\ref{lem:post_bel}, which we will use to further analyze the strongest possible attacker $\Astrong$ of Experiment~\ref{exp:ident}. 
\begin{lemma}[Calculation of the posterior belief]
	Assuming uniform priors and independent mechanisms $\cali{M}_i$ (more precisely, the noise of the mechanisms must be sampled independently), the posterior belief on dataset \cali{D} can be computed as
\begin{align*}
	\beta_k &= \frac{\prod_{i=1}^{k}\Pr(\cali{M}_{i}(\cali{D})=\vec{r}_i)}{\prod_{i=1}^{k}\Pr(\cali{M}_{i}(\cali{D})=\vec{r}_i)+\prod_{i=1}^{k}\Pr(\cali{M}_{i}(\cali{D'})=\vec{r}_i)}\\
	& = \frac{1}{1+\frac{\prod_{i=1}^{k}\Pr(\cali{M}_{i}(\cali{D'})=\vec{r}_i)}{\prod_{i=1}^{k}\Pr(\cali{M}_{i}(\cali{D})=\vec{r}_i)}}
\end{align*}
\label{lem:post_bel}
\end{lemma}
\iffullversion
\begin{proof} We prove the lemma by iteration over $k$.\\
$k=1$: We assume the attacker starts with uniform priors $\Pr(\cali{D})=\Pr(\cali{D'})=\frac{1}{2}$. Thus, $\beta_1(\cali{D}|R_1)$ can be directly calculated by dividing both numerator and denominator of $\beta$ by the numerator:
\begin{align*}
	\beta_1(\cali{D}|R_1) &=\frac{\Pr(\cali{M}_{1}(\cali{D})=\vec{r}_1)}{\Pr(\cali{M}_{1}(\cali{D})=\vec{r}_1)+\Pr(\cali{M}_{1}(\cali{D'})=\vec{r}_1)}\\
	& = \frac{1}{1+\frac{\Pr(\cali{M}_{1}(\cali{D'})=\vec{r}_1)}{\Pr(\cali{M}_{1}(\cali{D})=\vec{r}_1)}}
\end{align*}
$k-1\to k$: In the second step $\beta_{k-1}(\cali{D}|R_{k-1})$ is used as the prior, using the shorthand notations $\beta_k:=\beta_k(\cali{D}|R_k)$, and in the last step $p_k:=\Pr(\cali{M}_{k}(\cali{D})=\vec{r}_k)$ and $p_k':=\Pr(\cali{M}_{k}(\cali{D'})=\vec{r}_k)$ the calculation of $\beta_k(\cali{D}|R_k)$ starts as for the induction start $k=1$
\begin{align*}
	\beta_k & =\frac{\Pr(\cali{M}_{k}(\cali{D})=\vec{r}_k)\cdot\beta_{k-1}}{\Pr(\cali{M}_{k}(\cali{D})=\vec{r}_k)\cdot\beta_{k-1}+\Pr(\cali{M}_{k}(\cali{D'})=\vec{r}_k)\cdot(1-\beta_{k-1})}\\
	&=\frac{1}{1+\frac{\Pr(\cali{M}_{k}(\cali{D'})=\vec{r}_k)-\Pr(\cali{M}_{k}(\cali{D'})=\vec{r}_k)\cdot\beta_{k-1}}{\Pr(\cali{M}_{k}(\cali{D})=\vec{r}_k)\cdot\beta_{k-1}}} =\frac{1}{1+\frac{p_k'-p_k'\beta_{k-1}}{p_k \beta_{k-1}}}
\end{align*}

Now the induction assumption can be substituted for the right term of the denominator and then multiplying the numerator and denominator with $\prod_{i=1}^{k-1}p_i+\prod_{i=1}^{k-1}p_i'$ leads to 
\begin{align*}
    \frac{p_k'-p_k'\beta_{k-1}}{p_k \beta_{k-1}} &= \frac{p_k'-p_k'\frac{\prod_{i=1}^{k-1}p_i}{\prod_{i=1}^{k-1}p_i+\prod_{i=1}^{k-1}p_i'}}{p_k \frac{\prod_{i=1}^{k-1}p_i}{\prod_{i=1}^{k-1}p_i+\prod_{i=1}^{k-1}p_i'}} \\
    &= \frac{p_k'(\prod_{i=1}^{k-1}p_i+\prod_{i=1}^{k-1}p_i')-p_k' \prod_{i=1}^{k-1}p_i}
    {p_k \prod_{i=1}^{k-1}p_i}\\
    &= \frac{\prod_{i=1}^{k}p_i'}{\prod_{i=1}^{k}p_i}
\end{align*}
where in the last step the first and the third term in the denominator cancel and can be inserted back into the last form of $\beta_k$ above.
\end{proof}
\else \fi{}
In our analysis $\Astrong$ is a binary classifier that chooses the label with the highest posterior probability $\beta_k$. 
If prior beliefs are uniform, this decision process can be simplified. Consider $X_1:=\cali{M}(\cali{D}) \ \mathrm{and} \ X_0:=\cali{M}(\cali{D}')$. 
 Since $\Astrong$ knows $\cali{D}, \cali{D'}$ and $\cali{M}$, $\Astrong$ also knows the corresponding probability densities $g_{X_1}$ and $g_{X_0}$. The densities are identical and defined by $\cali{M}$, but are centered at the different results $f(\cali{D})$ and $f(\cali{D}')$, respectively, as visualized in Figure~\ref{fig:decision-boundary-pdfs} with $f(\cali{D})=0, f(\cali{D}')=1$. When $\Astrong$ has equal prior beliefs, $\Astrong$ decides whether $R_k$ is more likely to stem from $X_1$ or $X_0$ and therefore chooses 
\begin{equation} \label{eq:adversary}
\Astrong(R_k, \cali{D}, \cali{D'}, \cali{M}, \var{Dist}) = \argmax_{D\in\{\cali{D},\cali{D}'\}}{\beta(D|R_k)} = 
\argmax\limits_{b\in \{0,1\}}g_{X_b}(R_k)
\end{equation}

\begin{figure}
    \centering
	\begin{subfigure}{0.25\linewidth}
	\centering
		\includegraphics[width=1\linewidth]{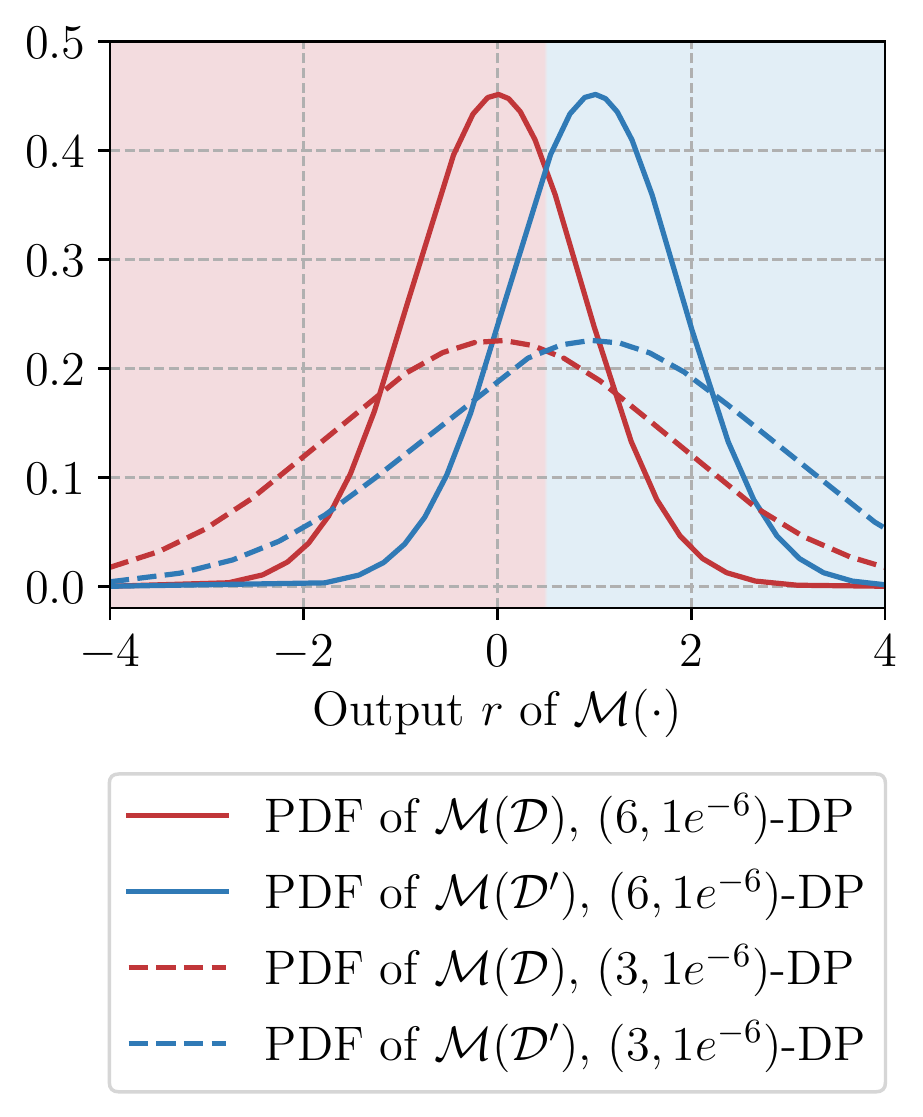}
		\caption{PDFs}
		\label{fig:decision-boundary-pdfs}
	\end{subfigure}%
	\begin{subfigure}{0.25\linewidth}
	\centering
		\includegraphics[width=1\linewidth]{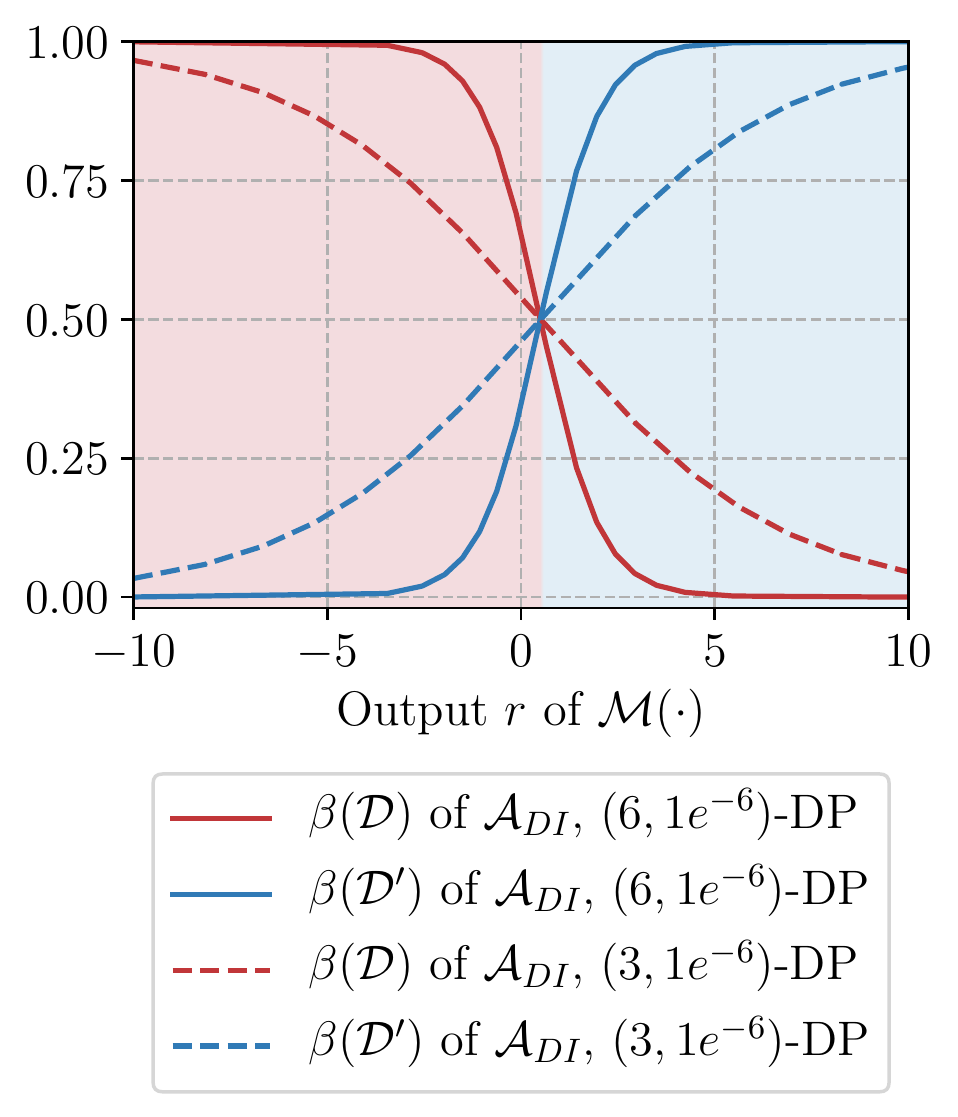}
		\caption{Posterior belief}
		\label{fig:decision-boundary-beliefs}
	\end{subfigure}%
	\caption{The decision boundary of $\Astrong$}
\end{figure}

$\beta(\cali{D})$ and $\beta(\cali{D'})$ for our example are visualized in Figure~\ref{fig:decision-boundary-beliefs}. $\Astrong$ acts as a naive Bayes classifier whose decision is depicted by the background color. The input features are the perturbed results $R_k$, and the exact probability distribution of each class is known. The distributions are entirely defined by $\cali{D}$, $\cali{D'}$, and $\cali{M}$, so $\Astrong$ does not use the knowledge of $\var{Dist}$. The posterior belief quantifies the probability of $R_k$; however, in another instance, $R_k$ could differ. In Section~\ref{subsec:adapt}, we will therefore define an upper bound on $\beta(\cali{D})$.

\subsection{Advantage in Identifying the Training Dataset}
\label{sec:adv}

The posterior belief $\beta_k$ quantifies the probability of inferring membership of a single record $\datapoint$. 
For example, when $\beta_k$ is low for a census dataset, the individual $\datapoint$ can plausibly deny presence in $\cali{D}$, and thus presence in the census.  
In practice, it is also important to know how often $\Astrong$ makes a correct guess, which only occurs when $\beta_k>0.5$. This is quantified by the advantage, which is the success rate normalized to the range $[-1,1]$, where $Adv=0$ corresponds to random guessing. Membership advantage was introduced to quantify the success of $\Ami$ \cite{YGF+18}; however, its definition can be used for $\Astrong$ of $\edi$. Generically: 

\begin{definition}[Advantage] \label{def:adv}
Given an experiment $Exp$ the advantage is defined as
\begin{equation*}
    Adv = 2\Pr(Exp=1)-1
\end{equation*}
where the probability is over the random iterative choices of the mechanisms up to step $k$. The advantage in $\edi$ is denoted $\madi$, while the advantage in $\emi$ is $\mami$. 
\end{definition}
\section{Derivation of upper bounds}
\label{sec:up_bound}
Within this section we use the DP guarantee to derive upper bounds for \textit{posterior belief} and \textit{advantage} in Sections~\ref{subsec:adapt} and~\ref{subsec:boundAdv}. In Section~\ref{sec:exp_ma}, we define expected membership advantage for the Gaussian mechanism, since the original bound is loose. 

\subsection{Upper Bound for the Posterior Belief} 
\label{subsec:adapt}

We formulate a generic bound on the Bayesian posterior belief that is independent of datasets $\cali{D}$ and $\cali{D'}$, the mechanism $\cali{M}$, and the result matrix $R=(\vec{r}_0,\vec{r}_1,\ldots,\vec{r}_k)$ comprising $k$ multidimensional mechanism outputs. The proposed bound solely assumes that the DP bound holds and makes no further simplifications, which results in an identifiability-based interpretation of DP guarantees. Theorem~\ref{cor:general-upper-bound} shows that $\Astrong$ operates under the sequential composition theorem, for both for $\eps$-DP and for $\epsdlt$-DP. 

\begin{theorem}[Bounds for the Adaptive Posterior Belief]
	\label{cor:general-upper-bound}
	Consider experiment $\edi$ with neighboring datasets $\cali{D}$ and $\cali{D'}$. Let	$\cali{M}_1,\ldots, \cali{M}_k$ be a sequence of arbitrary but independent differentially private learning algorithms.\\
	(i) Each $\cali{M}_i$ provides $\eps_1,\ldots,\eps_k$-DP to functions $f_i$ with multidimensional output.
	Then the posterior belief of $\Astrong$ is bounded by
	\begin{equation*}
	    \beta_k(\cali{D}|R_k) \le \rb = 
	    \frac{1}{1+e^{-\sum_{i=1}^{k} \eps_i }}
	\end{equation*}\\
     (ii) Each $\cali{M}_i$ provides $(\eps_i, \dlt_i)$-DP to multidimensional functions $f_i$. Then the same bound as above holds with probability $1-\sum_{i=1}^{k}\dlt_i$.
\end{theorem}
\iffullversion
\begin{proof}
(i) The adversary with unbiased prior (i.e., $0.5$) has a maximum posterior belief of $1/(1+e^{-\eps})$ when the \eps-differentially private Laplace mechanism is applied to a function with a scalar output~\cite{Lee2012}. This upper bound holds also for arbitrary \eps-differentially private learning algorithms with multidimensional output. We bound the general belief calculation by the inequality of Definition~\ref{def:DP}. Analogously, 
        $\Pr(\cali{M}(\cali{D})=\vec{r}) \leq e^\eps\Pr(\cali{M}(\cali{D'}) = \vec{r}) + \dlt$.
    Assuming equal priors, the posterior belief can be calculated as follows:
    \begin{align*}
    \beta(\cali{D}|R) &= \frac{1}{1+\frac{\prod_{i=1}^{k}\Pr(\cali{M}_{i}(\cali{D'})=\vec{r}_i)}{\prod_{i=1}^{k}\Pr(\cali{M}_{i}(\cali{D})=\vec{r}_i)}}\\ &
    \leq \frac{1}{1+\frac{\prod_{i=1}^{k}\Pr(\cali{M}_{i}(\cali{D'})=\vec{r}_i)}{\prod_{i=1}^{k}e^{\eps_i}\Pr(\cali{M}_i(\cali{D'}) = \vec{r}_i) + \dlt_i}}
    \end{align*}
    
    For $\dlt=0$, the last equation simplifies to: 
    \begin{align*}
        \beta(\cali{D}|R) & \leq \frac{1}{1+\frac{\prod_{i=1}^{k}\Pr(\cali{M}_{i}(\cali{D'})=\vec{r}_i)}{\prod_{i=1}^{k}e^{\eps_i}\Pr(\cali{M}_i(\cali{D'}) = \vec{r}_i) }} \\
        & = \frac{1}{1+\prod_{i=1}^{k}e^{-\eps_i} } = \frac{1}{1+e^{-\sum_{i=1}^{k} \eps_i }} = \rb
    \end{align*}
    
    (ii) We use properties of RDP to prove the posterior belief bound for multidimensional $(\eps_i, \dlt_i)$-differentially private mechanisms.
    \begin{align}
    	\beta_k(\cali{D}|R) & = \frac{1}{1+\prod_{i=1}^{k} \frac{\Pr(\cali{M}_{i}(\cali{D'})=\vec{r}_i)}{\Pr(\cali{M}_{i}(\cali{D})=\vec{r}_i)}} \label{eq:adaptive-gauss-basic}\\
    	& = \frac{1}{1 + \prod_{i=1}^{k} \frac{\Pr(\cali{M}_{i}(\cali{D'})=\vec{r}_i)}{\left(e^{\eps_{RDP,i}}\cdot\Pr(\cali{M}_i(\cali{D'}) = \vec{r}_i)\right)^{1-1/\alpha}}}\label{eq:prop10} \\
    	& = \frac{1}{1+\prod_{i=1}^{k} e^{-\eps_{RDP,i}(1-1/\alpha)}\cdot \Pr(\cali{M}_i(\cali{D'}) = \vec{r}_i)^{1/\alpha}} \label{eq:simplified_belief}
    	\end{align}
    	In the step from Eq.~\eqref{eq:adaptive-gauss-basic} to Eq.~\eqref{eq:prop10}, we use the probability preservation property,
    	    $\Pr(\mathcal{M}(\cali{D})=\vec{r}) \leq \left(e^{\epsilon_{RDP}}\Pr(\mathcal{M}(\cali{D'})=\vec{r})\right)^{1-1/\alpha}$, which appears in Langlois et al.~\cite{langlois2015} and generalizes Lyubashevsky et al.~\cite{Lyubashevsky2013}.
    	This same property was used by Mironov~\cite{Mironov2017} to prove that RDP guarantees can be converted to \epsdlt~guarantees. In the context of this proof, Mironov also implies that $\eps$-DP holds when $e^{\eps_{RDP}}\Pr(\cali{M}(\cali{D'}) = \vec{r})>\dlt^{\alpha/(\alpha-1)}$, since otherwise $\Pr(\cali{M}(\cali{D}) = \vec{r})\le\dlt$. We therefore assume $e^{\eps_{RDP}}\Pr(\cali{M}(\cali{D'}) = \vec{r})>\dlt^{\alpha/(\alpha-1)}$, which occurs in at least $1-\dlt$ cases. We continue from Eq.~\eqref{eq:simplified_belief}:
    	\begin{align}
    		\beta_k(\cali{D}) & \le \frac{1}{1+ \prod_{i=1}^{k} e^{-\eps_{RDP,i}(1-1/\alpha)}\cdot\left(\dlt_i^{\alpha/(\alpha-1)}\cdot e^{-\eps_{RDP,i}}\right)^{1/\alpha}} \nonumber \\
    		& = \frac{1}{1+ \prod_{i=1}^{k} e^{-\eps_{RDP,i}}\cdot\dlt_i^{1/(\alpha-1)}} \nonumber \\
    		& = \frac{1}{1+ \prod_{i=1}^{k} e^{-\eps_{RDP,i}}\cdot e^{-1/(\alpha-1)\ln
    (1/\dlt_i)}} \nonumber \\
    		& = \frac{1}{1 + \prod_{i=1}^{k} e^{-(\eps_{RDP,i}+(\alpha-1)^{-1}\ln
    (1/\dlt_i))}} \label{eq:simple_rdp_belief}\\
    	& = \frac{1}{1+\prod_{i=1}^{k} e^{-\eps_{i}}} = \frac{1}{1+e^{-\sum_{i=1}^{k}\eps_{i}}} = \rb \label{eq:adaptive-gauss-end}
    \end{align}
\end{proof}
\else \fi{}

Note that we use the conversion from RDP to DP in the step from Eq.~\ref{eq:simple_rdp_belief} to Eq.~\ref{eq:adaptive-gauss-end} (cf.~Section~\ref{sec:dp}). Equivalently one can specify a desired posterior belief and calculate the overall \eps, which can be spent on a composition of differentially private queries:
\begin{equation} \label{eq:epsFromBound}
    \eps = \ln\left(\frac{\rb}{1-\rb}\right)
\end{equation}

The value for \dlt can be chosen independently according to the recommendation that $\dlt\ll\frac{1}{N}$ with $N$ points in the input dataset~\cite{dwork2014}.

\subsection{Upper Bound for the Advantage in Identifying the Training Dataset for General Mechanisms}
\label{subsec:boundAdv}

We now formulate an upper bound for the advantage $\madi$ of $\Astrong$ in Proposition~\ref{prop:generalBoundAdv}. The membership advantage of $\Ami$ has been bounded in terms of $\eps$ and defines $\Ami$'s success~\cite{YGF+18}. The general bound for $\Ami$ also holds for $\Astrong$ based on Proposition~\ref{prop:reduction}. 

\begin{proposition}[Bound on the Expected Membership Advantage for $\Astrong$]\label{prop:generalBoundAdv}
     For any $\eps$-DP mechanism the identification advantage of $\Astrong$ in experiment $\edi$ can be bounded as
\begin{equation*}
\madi \leq(e^\eps-1) \Pr(\Astrong=1|b=0) 
\end{equation*}
\end{proposition}
\iffullversion
\begin{proof}
First the definition is rewritten by separating true positives and true negatives. Then using that both datasets are equally likely to be chosen by the challenger ($\Pr(b=1)=\Pr(b=0) = 1/2$). We substitute $\Pr(b'=0 \ \vert b=0)$ by the probability of the complementary event $1-\Pr(b'=1 \ \vert b=0))$ and $b'=1$ by $\Astrong=1$, which leads to Eq.~\eqref{eq:TPminusFN}‚ of Yeom et al.~\cite{YGF+18}
\begin{align}
	\madi &= 2(\Pr(b=1)\Pr(b'=1 \ \vert b=1) \nonumber\\
	& \hphantom{=} +\Pr(b=0)\Pr(b'=0 \ \vert b=0))-1 \nonumber\\
	&=\Pr(\Astrong=1 | b=1) - \Pr(\Astrong=1|b=0) 
	\label{eq:TPminusFN}
\end{align}
which is the difference between the probability for detecting $\cali{D}$ and the probability of incorrectly choosing $\cali{D}$. Now we use the fact that the mechanism $\cali{M}$ turns $r$ into random variables $X_1:=\cali{M}(\cali{D})$ and $X_0:=\cali{M}(\cali{D'})$ for the cases $b=1$ and $b=0$, respectively. We formulate the probability density functions as $g_{X_1}$ and $g_{X_0}$. Additionally $A(r)$ is introduced as a shorthand for $\Astrong(\vec{r}, \cali{D}, \cali{D'}, \cali{M}, \var{Dist})$

\begin{align}
	\madi &=\Pr(\Astrong=1 | r=\cali{M}(\cali{D})) - \Pr(\Astrong=1|r=\cali{M}(\cali{D'})) \nonumber\\
&=\mathbb{E}_{r=\cali{M}(\cali{D})}
 (\Astrong(\vec{r}, \cali{D}, \cali{D'}, \cali{M}, \var{Dist})) - \nonumber
	\\ & \hphantom{=} \mathbb{E}_{r=\cali{M}(\cali{D'})} (\Astrong(\vec{r}, \cali{D}, \cali{D'}, \cali{M}, \var{Dist})) \nonumber\\
& = \int g_{X_1}(\vec{r}) A(\vec{r})\mathrm{d}\vec{r} -\int g_{X_0}(\vec{r}) A(\vec{r}) \mathrm{d}\vec{r} \label{eq:diffOfDistr}\\
& = \int (g_{X_1}(\vec{r})-g_{X_0}(\vec{r})) A(\vec{r}) \mathrm{d}\vec{r} \label{eq:distrDiff}
\end{align}

Since $\eps$-DP formulated as $Pr(\cali{M}(\cali{D})\in S)\leq e^\eps Pr(\cali{M}(\cali{D'})\in S)$ holds for all $S$ it yields the same inequality $g_{X_1} \leq e^\eps g_{X_0}$ for the densities at each point
\begin{align}
	\madi	&\leq (e^\eps-1)\int g_{X_0}(\vec{r}) A(r) \mathrm{d}\vec{r} \nonumber \\
	    & = (e^\eps-1) \Pr(\Astrong=1|b=0) \nonumber \\
	    & \leq e^\eps - 1 \nonumber
\end{align}
\end{proof}
\else \fi{}
Bounding $\Pr(\Astrong=1|b=0)$ by 1 results in $\madi \leq e^\eps - 1$. When $\Astrong$ acts as a naive Bayes classifier, only a complete lack of utility from infinite noise results in $\Pr(\Astrong=1|b=0)=0.5$. Otherwise, $\Pr(\Astrong=1|b=0)\ll0.5$; therefore, the membership advantage bound is usually not tight. This is in line with Jayaraman et al.~\cite{JE19} who expect that this would be the case for MI.

\subsection{Upper Bound for the Advantage in Identifying the Training Dataset for Gaussian Mechanisms}
\label{sec:exp_ma}
In practice, $\Astrong$ will be faced with a specific DP mechanism, and we focus on the mechanism used in DPSGD to find a tighter bound than the generic bound described in the previous section.
We use the notation $\AstrongGau$ and $\madiGau$ to specify the adversary and advantage of an instantiation of $\Astrong$ against the Gaussian mechanism with \epsdlt-DP. 
\iffullversion We now derive a tighter bound $\ra$ on $\madiGau$ and continue from Eq.~\eqref{eq:distrDiff}. \else We now derive a tighter bound $\ra$ on $\madiGau$. \fi{} Note that under the assumption of equal priors, the strongest possible adversary of Eq.~\eqref{eq:adversary} \iffullversion maximizes Eq.~\eqref{eq:distrDiff} by choosing \else chooses \fi{} $b=1$ if $(g_{X_1}(\vec{r})-g_{X_0}(\vec{r}))>0$ and $b=0$ otherwise. The resulting bound on $\madiGau$ is constructed from $\AstrongGau$'s strategy; however, the bound holds for all weaker adversaries, including $\Ami$. Since we argue that $\AstrongGau$ precisely represents the assumptions of DP, the bound should hold for other possible attacks in the realm of DP and the Gaussian mechanism under the i.i.d.~assumption.

Since $\AstrongGau$ is a naive Bayes classifier with known probability distributions, we use the properties of normal distributions (we refer to Tumer et al.~\cite{Tumer1996} for full details). We find that the decision boundary does not change under $\cali{M}_{Gau}$ with different \epsdlt~guarantees as long as the probability density functions (PDF) are symmetric. Holding $\cali{M}(\cali{D})=r$ constant and reducing \epsdlt~solely affects the posterior belief of $\AstrongGau$, not the choice of $\cali{D}$ or $\cali{D'}$. For example, consider the example of Figure~\ref{fig:compare_distributions_epsilons}. If a $(6,10^{-6})$-DP $\cali{M}_{Gau}$ is applied for perturbation, $\AstrongGau$ has to choose between the two PDFs in Figure~\ref{fig:error-rate-pdfs-eps6}. Increasing the privacy guarantee to $(3,10^{-6})$-DP in Figure~\ref{fig:error-rate-pdfs-eps3} squeezes the PDFs and belief curves. The corresponding regions of error are shaded in Figures~\ref{fig:error-rate-pdfs-eps6} and~\ref{fig:error-rate-pdfs-eps3}, where we see that a stronger guarantee reduces $\madiGau$. 
\begin{figure}
    \centering
	\begin{subfigure}{0.25\linewidth}
		\includegraphics[width=1\linewidth]{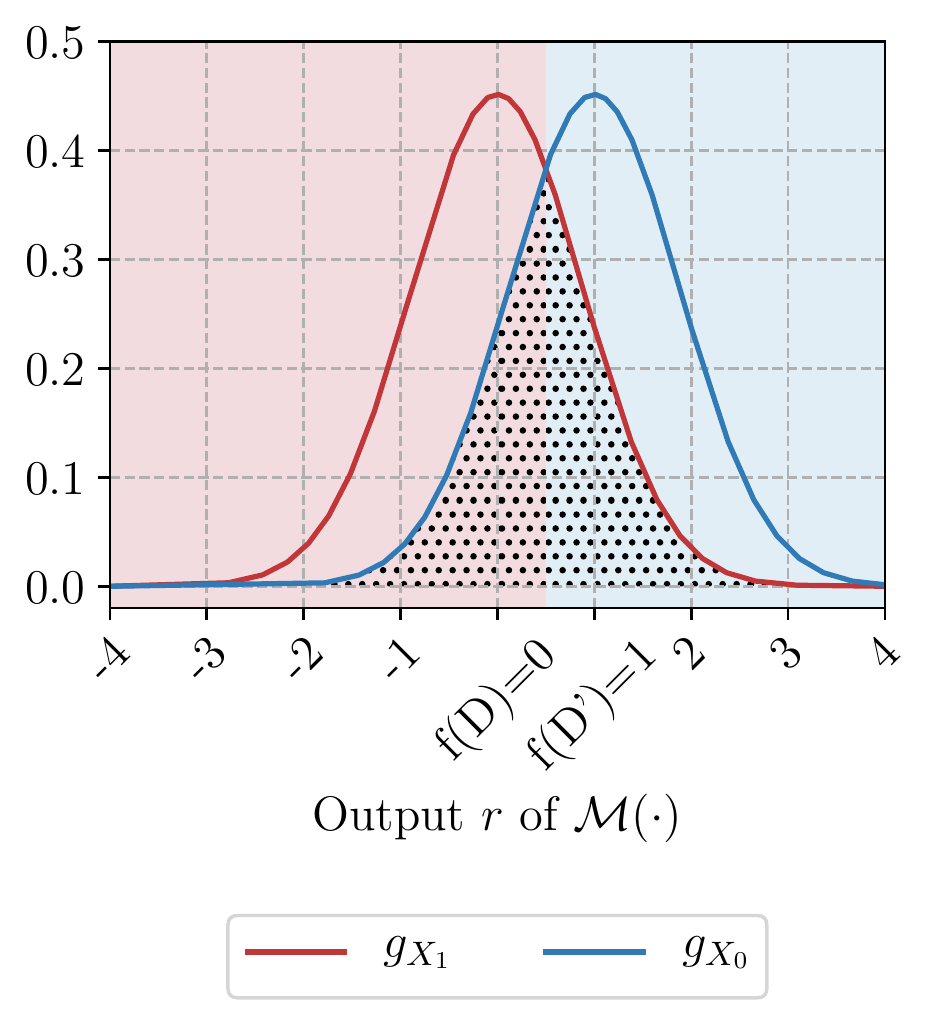}
		\caption{Error regions $(6,1e^{-6})$-DP}
		\label{fig:error-rate-pdfs-eps6}
	\end{subfigure}%
	\begin{subfigure}{0.25\linewidth}
		\includegraphics[width=1\linewidth]{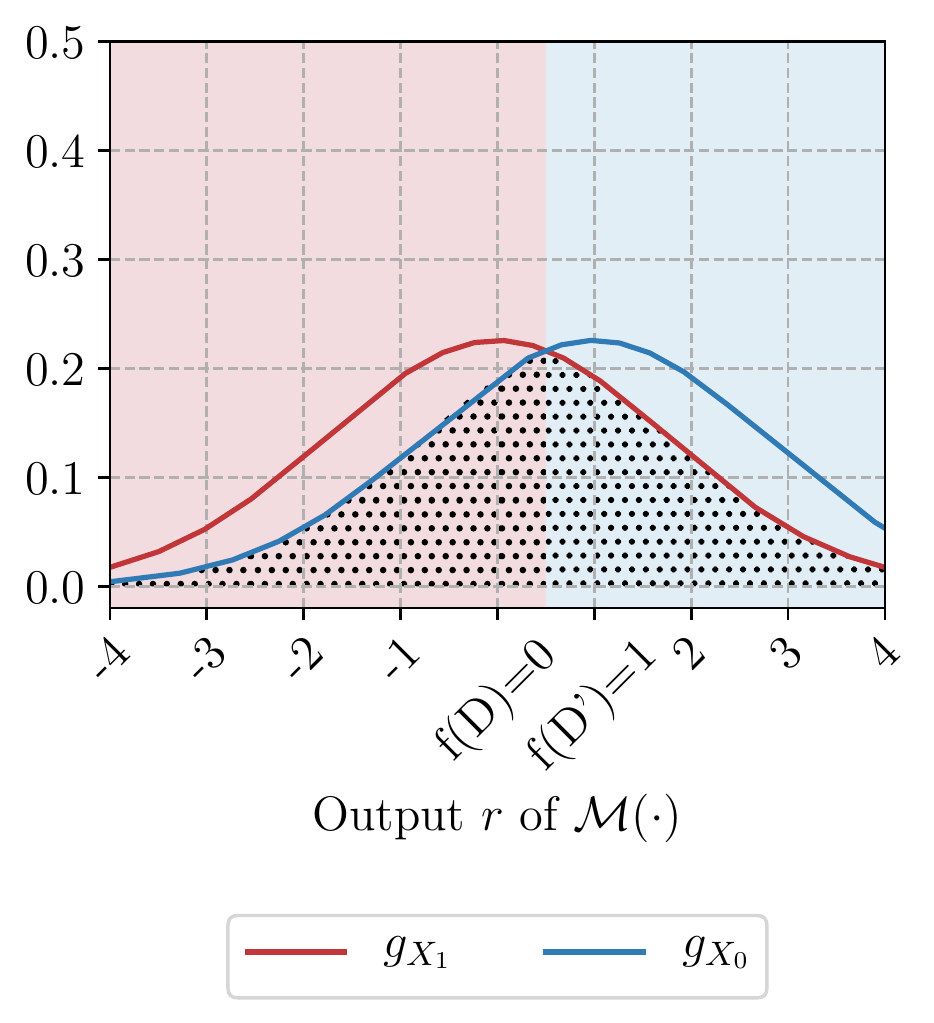}
		\caption{Error regions $(3,1e^{-6})$-DP}
		\label{fig:error-rate-pdfs-eps3}
	\end{subfigure}%
	\caption{error regions for varying \eps,~$\cali{M}_{Gau}$}
	\label{fig:compare_distributions_epsilons}
\end{figure}
 
We assume throughout this paper that $\AstrongGau$ has uniform prior beliefs on the possible databases $\cali{D}$ and $\cali{D'}$. This distribution is iteratively updated based on the posterior resulting from the mechanism output $r$. If $\cali{M}_{Gau}$ is used to achieve \epsdlt-DP, we can determine the expected membership advantage of the practical attacker $\AstrongGau$ analytically by the overlap of the resulting Gaussian distributions~\cite[p.~321]{MKB79}.
We thus consider two multidimensional Gaussian PDFs (i.e., $\cali{M}(\cali{D})$, $\cali{M}(\cali{D'})$) with covariance matrix $\Sigma$ and means (without noise) $\vec{\mu_1}=f(\cali{D}),\vec{\mu_2}=f(\cali{D'})$. This leads us to Theorem~\ref{def:expected-adversarial-success-rate}. 
\begin{theorem}[Tight Bound on the Expected Adversarial Membership Advantage]
	\label{def:expected-adversarial-success-rate}
	For the \epsdlt-differentially private Gaussian mechanism, the expected membership advantage of the strong probabilistic adversary on either dataset $\cali{D}, \cali{D'}$. 
	\begin{align*}
	    \madi \leq \ra = 2\Phi\left(\frac{\eps}{2 \sqrt{2\ln
(1.25/\dlt)}}\right) - 1
	\end{align*}
	where $\Phi$ is the cumulative density function of the standard normal distribution.
\end{theorem}
\iffullversion
\begin{proof} 
We start from Eq.~\eqref{eq:diffOfDistr} where the Gauss-distributions are $g_{X_1}$ and $g_{X_0}$. Since both distributions arise from the same mechanism they have the same $\Sigma$ but different means $\mu_1=f(\cali{D})$ and $\mu_0=f(\cali{D'})$. Since the strongest adversary is the Bayes adversary that chooses according to Eq.~\eqref{eq:adversary} and we assume equal priors, the decision boundary between $\cali{D}$ and $\cali{D}'$ is the point of intersection of the densities (see Figure~\ref{fig:error-rate-pdfs-eps6} for the 1D-case). We use linear discriminant analysis where the boundary is a hyperplane halfway between $\mu_1$ and $\mu_0$. This plane is halfway ($\Delta /2$) between the two centers, where $\Delta$ is the Mahalanobis distance~\cite{Maha1936}	$\Delta = \sqrt{(\vec{\mu_1}-\vec{\mu_2})^T\Sigma^{-1}(\vec{\mu_1}-\vec{\mu_2})}$. Notably the decision boundary between \cali{D} and \cali{D'} does not depend on $\Sigma$, but the possible distance between $\mu_1$ and $\mu_0$ (i.e., sensitivity). As we add independent noise in all dimensions $\Sigma = \sgm^2\mathbb{I}$, we simplify all calculations from Eq.~\eqref{eq:diffOfDistr} to the one-dimensional case and simplify $\Delta=\frac{\|\vec{\mu_1} - \vec{\mu_2}\|_2}{\sgm}$. 
Thus,
    \begin{align} \label{eq:advantageGaussSingle}
		\madiGau &= \Phi(\Delta/2)-\Phi(-\Delta/2) \nonumber = 2\Phi(\Delta/2)-1 \nonumber \\
		& = 2\Phi\left(\frac{\|\vec{\mu_1} - \vec{\mu_2}\|_2}{2\sgm}\right) -1 
	\end{align}
Inserting the standard deviation needed for $\epsdlt$-DP from Eq.~\eqref{eq:sigmaEpsDlt} then yields
   \begin{align*}
		\madiGau &= 
		2\Phi\left(\frac{\|\vec{\mu_1} - \vec{\mu_2}\|_2}{2GS_{f_2}(\sqrt{2\ln
(1.25/\dlt)}/\eps)}\right) -1 \\
		& \le 2\Phi\left(\frac{\eps}{2(\sqrt{2\ln
(1.25/\dlt)})}\right) - 1 = \ra
	\end{align*}
\end{proof}
\else\fi{}

We can calculate $\eps$ from a chosen maximum expected advantage
\begin{equation}\label{eq:adv_to_eps}
\eps = \sqrt{2\ln
(1.25/\dlt)} \, \Phi^{-1}\left(\frac{\ra+1}{2}\right)
\end{equation}
\epsdlt guarantees with $\dlt>0$ can be expressed via a scalar value $\ra$. Summarizing, we now have complementary interpretability scores, where $\rb$ represents a bound on individual deniability and $\ra$ relates to the expected probability of reidentification. 
While $\rb$ holds for all mechanisms, $\ra$ was derived solely for the Gaussian mechanism.
We provide example plots of $\rb$ and $\ra$ for different \epsdlt in Figure~\ref{fig:general-confidence-successrate-curves}. To compute both scores, we use Theorems~\ref{cor:general-upper-bound} and~\ref{def:expected-adversarial-success-rate}. We set $f(\cali{D}) = (0_1,0_2,\ldots,0_k)$ and $f(\cali{D'}) = (1_1,1_2,\ldots,1_k)$ for all dimensions $k$, so $GS_{f_2} = \sqrt{k}$. Figure~\ref{fig:expected-confidence-bound}~illustrates that there is no significant difference for $\rb$ between \eps-DP and \epsdlt-DP. In contrast, $\ra$ strongly depends on the choice of \dlt.

\begin{figure}
	\centering
	\begin{subfigure}{0.25\linewidth}
		\centering
		\includegraphics[width=1\linewidth]{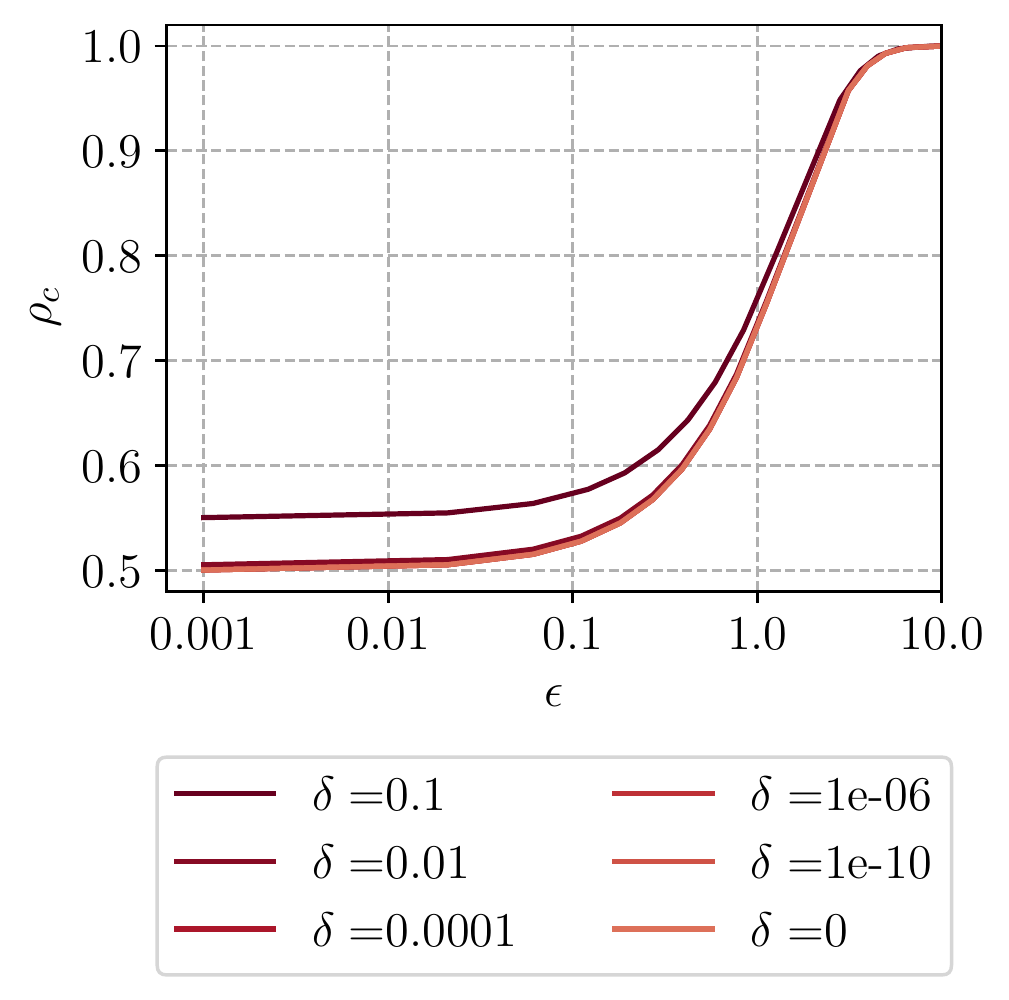}
		\caption{$\rb$}
		\label{fig:expected-confidence-bound}
	\end{subfigure}%
	\begin{subfigure}{0.25\linewidth}
		\centering
		\includegraphics[width=1\linewidth]{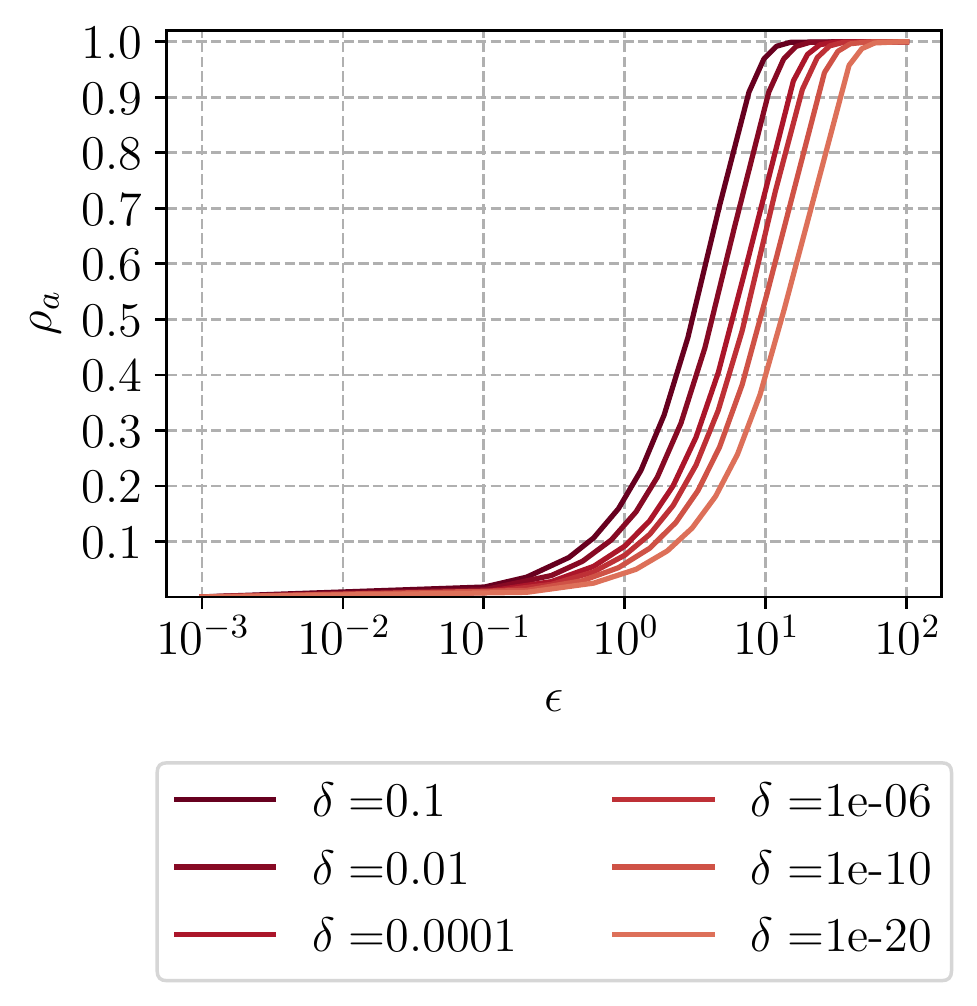}
		\caption{$\ra$}
		\label{fig:expected-success-rate}
	\end{subfigure}%
	\caption{$\rb$ and $\ra$ for various \epsdlt when using $\cali{M}_{Gau}$}
	\label{fig:general-confidence-successrate-curves}
\end{figure}

\subsection{RDP Instead of Sequential Composition}
\label{subsec:comp}

In iterative settings, such as NN training, the data scientist will have to perform multiple mechanism executions, which necessitates the use of composition theorems to split the total guarantee into guarantees per iteration $(\eps_i, \dlt_i)$. Sequential composition only offers loose bounds in practice~\cite{dwork2010, KOV17}; we suggest using RDP composition, which allows a tight analysis of the privacy loss over a series of mechanisms. Therefore, we adapt both $\rb$ and $\ra$ to RDP.

We first demonstrate that RDP composition results in stronger \epsdlt guarantees than sequential composition for a fixed bound $\rb$. \iffullversion We start from Eq.~\eqref{eq:simple_rdp_belief}:
\else We start from the simple RDP belief:\fi{}

\begin{align}
	\beta_k(\cali{D}|R) & \le \frac{1}{1+\prod_{i=1}^{k} e^{-(\eps_{RDP,i}+(\alpha-1)^{-1}\ln
(1/\dlt_i))}} \nonumber \\
	& = \frac{1}{1+e^{k(\alpha-1)^{-1}\ln
(\dlt_i)-\sum_{i=1}^{k}\eps_{RDP,i}}} \label{eq:rdp-delta-out} \\
	& = \frac{1}{1+e^{(\alpha-1)^{-1}\ln
(\dlt_i^k)-\sum_{i=1}^{k}\eps_{RDP,i}}} \nonumber \\
	& = \frac{1}{1+e^{-(\sum_{i=1}^{k}\eps_{RDP,i}-(\alpha-1)^{-1}\ln
(\dlt_i^k))}} = \rb \label{eq:rdp-eps-composition}
\end{align}

We assume the same value of $\dlt_i$ is used during every execution and can therefore remove it from the sum in Eq.~\eqref{eq:rdp-delta-out}. Eq.~\eqref{eq:rdp-eps-composition} and the conversion ($\alpha$, $\eps_{RDP}$)-RDP to $(\eps_{RDP}-\frac{\ln\dlt}{\alpha-1},\dlt)$-DP imply that an RDP-composed bound can be achieved with a composed \dlt equal to $\dlt_i^k$. We know that sequential composition results in a composed \dlt value equal to $k\dlt_i$. Since $\dlt^k<k\dlt$, RDP offers a stronger \epsdlt guarantee for the same $\rb$, and results in a tighter bound for $\rb$ under composition. This behavior can also be interpreted as the fact that holding the composed \epsdlt guarantee constant, the value of $\rb$ is greater when sequential composition is used compared to RDP.

A similar analysis of the expected membership advantage under composition is required when considering a series of mechanisms $\cali{M}$. We restrict our elucidations to the Gaussian mechanism. The $k$-fold composition of $\cali{M}_{Gau_i}$, each step guaranteeing $(\alpha, \eps_{RDP,i})$-RDP, can be represented by a single execution of $\cali{M}_{Gau}$ with $k$-dimensional output guaranteeing $(\alpha,\eps_{RDP} = k\eps_{RDP,i})$-RDP. \iffullversion We start from Eq.~\eqref{eq:advantageGaussSingle}, and use Eq.~\eqref{eq:gaussrdp} and the fact that $GS_{f_2}$ bounds $\|\mu_{1,i}-\mu_{2,i}\|$.
\else We use Eq.~\eqref{eq:gaussrdp} and the fact that $GS_{f_2}$ bounds $\|\mu_{1,i}-\mu_{2,i}\|$. \fi{}

\begin{align*}
	\madiGau 
    & = 2\Phi\left(\frac{\|\vec{\mu_1} - \vec{\mu_2}\|_2}{2\sgm_i}\right) -1  = 2\Phi\left(\frac{\sqrt{k}\|\mu_{1,i} - \mu_{2,i}\|_2}{2GS_{f_2} \sqrt{\alpha/(2\eps_{RDP,i})}}\right) -1\\
    & \leq 2\Phi\left(\frac{\sqrt{k}}{2 \sqrt{\alpha/(2\eps_{RDP,i})}}\right) -1  = 2\Phi\left(\sqrt{\frac{k\eps_{RDP,i}}{2\alpha}}\right) - 1 \\
    & = 2\Phi\left(\sqrt{\frac{\eps_{RDP}}{2\alpha}}\right) -1 = \ra \label{eq:rdp-rho-alpha}
\end{align*}

The result shows that $\AstrongGau$ fully takes advantage of the RDP composition properties of $\eps_{RDP,i}$ and $\alpha$; as expected, $\ra$ takes on the same value, regardless of whether $k$ composition steps with $\eps_{RDP,i}$ or a single composition step with $\eps_{RDP}$ is carried out. Therefore, we can calculate the final $\ra$ for functions with multiple iterations, such as the training of deep learning models, and $\ra$ can be decomposed into a privacy guarantee per composition step with RDP. 

\section{Application to Deep Learning}
\label{sec:tight}

In DPSGD, the stochastic gradient descent optimizer adds Gaussian noise with standard deviation \sgm to the computed gradients. The added noise ensures that the learned NN is \epsdlt differentially private w.r.t.~the training dataset. This section illustrates our method for choosing DPSGD privacy parameters. Data scientists may first choose upper bounds for the posterior belief, from which \eps is obtained using Eq.~\eqref{eq:epsFromBound}. From \eps and the sensitivity, the standard deviation \sgm of the Gaussian noise is determined.

We discuss a heuristic for estimating the local sensitivity in Section~\ref{subsec:settingDPSGD}. Then, Section~\ref{subsec:quantifyIdentifiability} formulates an algorithm for implementing $\AstrongGau$, and discusses how this algorithm is used to empirically quantify the posterior belief and the advantage. Finally, using the implemented adversary $\AstrongGau$ a method for auditing the privacy loss \eps and the bounds derived in Section~\ref{sec:up_bound} is provided in Section~\ref{subsec:auditEps}.

\subsection{Setting Privacy Parameters and Determining the Sensitivity}
\label{subsec:settingDPSGD}
Based on the recommendation to set \clip to the median of the norms of unclipped gradients~\cite{abadi2016} we set $\clip=3$ in all our experiments. In the following, we describe how to set up the system in order to determine the standard deviation of Gaussian noise \sgm. We want to limit $\AstrongGau$'s belief of distinguishing a training dataset differing in any chosen person by setting the upper bound for the posterior belief $\rb$. We then transform $\rb$ to an overall \eps for the $k$ update steps in DPSGD using Eq.~\eqref{eq:epsFromBound}, 
which in turn leads to \sgm for the DPSGD using Eq.~\eqref{eq:sigmaEpsDlt}.
In Eq.~\eqref{eq:sigmaEpsDlt} two parameters need to be set: $\sens$ and $\dlt$. While we set $\dlt$ to $1/|\cali{D}|$ for all experiments, the choice of $\sens$ is more challenging. The upper bound for the privacy loss \eps can only be reached when $\sens$ is set specifically to the sensitivity of the dataset at hand. We can calculate the local sensitivity for bounded DP as 
\begin{equation} \label{eq:sensitivityADI}
    LS_{\clgrad_i}(\cali{D}) = n\cdot||\clgrad_i(\cali{D'})-\clgrad_i(\cali{D})\| \nonumber,
\end{equation}

and for unbounded DP as
\begin{equation} 
    LS_{\clgrad_i}(\cali{D}) = ||(n-1)\cdot\clgrad_i(\cali{D'})-n\cdot\clgrad_i(\cali{D})\| \nonumber,
\end{equation}

where $\clgrad_i(\cali{D})$ and $\clgrad_i(\cali{D'})$ represent the average of all clipped, unperturbed per-example gradients $\bar{g}_i(x) \forall x\in\cali{D}$ and $x\in\cali{D'}$, respectively.

Since clipping is done before perturbation, the global sensitivity $GS_f$ in DPSGD is set to the clipping norm for unbounded DP, i.e., $GS_{f}=\clip$. The sensitivity bounds the impact of a data point on the total gradient, equivalent to the difference between the gradients differing between $\cali{D}$ and $\cali{D'}$, which is artificially bounded by $\clip$ for unbounded DP. For bounded DP where one record is instead replaced with another in $\cali{D'}$, the lengths of the clipped gradients of these two records could each be \clip and point in opposite directions resulting in $n\cdot||\clgrad_i(\cali{D'})-\clgrad_i(\cali{D})\|_2\leq2\clip$.

Although \clip bounds the influence of a single training record on the gradient, \clip may well be loose, since \clip does not necessarily reflect the factual difference between the training dataset and possible neighboring datasets. When \clip is loose, the DP bound on privacy loss \eps is not reached, and the identifiability metrics $\ra$ and $\rb$ will not be reached either. Nissim et al.~\cite{Nissim2007} proposed local sensitivity $LS_f$ to specifically scale noise to the input data. The use of $LS_f$ decreases the noise scale by narrowing the DP guarantee from protection against inference on any possible adjacent datasets to inference on the original dataset and any adjacent dataset. In ML projects training and test data are often sampled from a static holdout, where all data points stem from a domain of similar data. If the holdout is a very large dataset, only the specific neighboring datasets possible in this domain need to be protected under DP. To reach the DP bound, we suggest fixation of the training dataset \cali{D} and considering only neighboring datasets \cali{D'} adjacent to \cali{D}.

However, approximating $LS_{\clgrad_i}$ for NN training is difficult because the gradient function output depends not only on \cali{D} and \cali{D'}, but also on the architecture and current weights of the network. To ease this dilemma, we propose \textit{dataset sensitivity} in Definition~\ref{def:dataset_sensitivity}. Dataset sensitivity is a heuristic with which we strive to consider the neighboring dataset $\hat{\cali{D}'}$ with the largest difference to \cali{D} within the overall ML dataset \cali{U} in an effort to approximate $LS_{\clgrad_i}$. We assume that similar data points will result in similar gradients. While this assumption does not necessarily hold under crafted adversarial examples~\cite{GSS15}, for which privacy protection cannot be guaranteed, the malicious intent renders the necessity for their protection debatable.
In Definition~\ref{def:dataset_sensitivity} the dissimilarity measure of specific datasets is not further specified. 

\begin{definition}[Dataset Sensitivity]
\label{def:dataset_sensitivity}

Consider a given dataset \cali{U}, a training dataset \cali{D} $\subseteq$ \cali{U}, all neighboring datasets \cali{D'} $\subseteq$ \cali{U} and a dissimilarity measure $d$. The dataset sensitivity $DS(\cali{D)}$ w.r.t.~dissimilarity measure $d$ is then defined as 
\begin{equation*}
DS(\cali{D}) = \max\limits_{\cali{D}'}d(\cali{D},\cali{D}')
\end{equation*}
and consequently
\[
\hat{\cali{D}'}:=\argmax\limits_{\cali{D}'}d(\cali{D},\cali{D}')
\]
\end{definition}

In practice, if a dissimilarity or distance measure $d$ of individual data points is available, it can be used to find the most dissimilar neighboring dataset $\hat{\cali{D'}}$ that maximizes the dataset sensitivity. The computation of $\cali{D}'$ depends on the neighboring datasets and is different for unbounded and bounded DP. More precisely, for unbounded DP one forms $\hat{\cali{D}'}=\cali{D} \setminus \{\datapoint'\}$ by removing the most dissimilar data point $\hat{\datapoint}$ from the training data
\begin{equation} \label{eq:xunbounded}
\hat{\datapoint}=\argmax\limits_{\datapoint_1\in \cali{D}}\sum_{\datapoint_2\in\cali{D}\setminus \datapoint_1} d(\datapoint_1,\datapoint_2)
\end{equation}
The dataset $\hat{\cali{D}'}$ is then used to approximate the local sensitivity $LS_{\clgrad_i}$ by
\begin{equation}\label{eq:datapointLSunbounded}
LS_{\clgrad_i}(\cali{D})\approx \hat{LS}_{\clgrad_i}(\cali{D}):= \|\bar{g}_i(\hat{\datapoint})\|,
\end{equation}
where $\bar{g}_i(\datapoint)$ is the clipped gradient of data point $\datapoint$ in step $i$.
The simplification from $LS_{\clgrad_i}$ to $DS$ allows us to bypass the complex gradient calculations to identify dissimilar \cali{D} and \cali{D'}. The computational complexity of computing the dataset sensitivity only depends on the dataset size $n$, but not the number of iterations $k$, like the local sensitivity does. For bounded DP where a neighboring dataset is formed by replacing an element $\{\datapoint\}\in \cali{D}$ with an element $\datapoint' \in \cali{U} \setminus \cali{D}$ one searches for 
\begin{equation} \label{eq:xbounded}
(\hat{\datapoint},\hat{\datapoint'})=\argmax\limits_{\datapoint \in \cali{D}, \datapoint' \in \cali{U} \setminus \cali{D}}d(\datapoint,\datapoint').
\end{equation}
and approximates the local sensitivity as
\begin{equation} 
\label{eq:datapointLS}
LS_{\clgrad_i}(\cali{D})\approx \hat{LS}_{\clgrad_i}(\cali{D}):= \|\bar{g}_i(\hat{\datapoint})-\bar{g}_i(\hat{\datapoint'})\|
\end{equation}

\subsection{Empirical Quantification of Posterior Beliefs and Advantages}
\label{subsec:quantifyIdentifiability}
In Section~\ref{subsec:settingDPSGD} the noise scale \sgm limits the upper bound for the posterior belief of $\Astrong$ on the original dataset $\cali{D}$. According to Theorem~\ref{cor:general-upper-bound} this upper bound holds with probability $1-\dlt$. For a given dataset, the posterior belief might be much smaller than the bound, so it is desirable to determine the empirical posterior belief on $\cali{D}$. The same holds for the advantage $\madi$ and the upper bound $\ra$ from Theorem~\ref{def:expected-adversarial-success-rate} w.r.t.~identifying dataset $\cali{D}$. 
We formulate an implementation of the adversary $\AstrongGau$ which allows us to assess the empirical posterior belief $\beta$ and membership advantage $\madi$, and thus the empirical privacy loss of specific trained models.

The adversary $\AstrongGau$ strives to identify the training dataset, having the choice between neighboring datasets \cali{D} and \cali{D'}.
In addition to \cali{D} and \cali{D'}, $\AstrongGau$ is assumed to have knowledge of the NN learning parameters and updates after every training step $i\leq k$: learning rate $\eta$, weights $\theta_i$, perturbed gradients $\dpgrad_i$, privacy mechanism $\cali{M}_i$, parameters $\epsdlt$, $\var{C}$, the resulting standard deviation $\sigma$ of the Gaussian distribution and the prior beliefs. The implementation of $\Astrong$ for DPSGD is provided in Algorithm~\ref{alg:adversary_algo}. 

In each learning step $\Astrong$ first computes the unperturbed, clipped batch gradients for both datasets based on the resulting weights from the previous step of the perturbed learning algorithm (Steps~\ref{step:dgrad} and~\ref{step:d_prime_grad}). Then $\AstrongGau$ calculates the sensitivity. The $\eps_i$ and $\delta_i$ for each iteration is calculated using RDP composition (cf.~Eq.~\eqref{eq:gaussrdp}). Consequently, the Gaussian mechanism scale \sgm is calculated from \epsdlt and \sens using Eq.~\eqref{eq:sigmaEpsDlt}. 
Using the standard deviation \sgm, the posterior belief $\beta_i$ is updated in Step~\ref{algstep:belief} based on the observed perturbed clipped gradient $\dpgrad_i$ and the unperturbed gradients from Steps~\ref{step:dgrad} and ~\ref{step:d_prime_grad}. The calculation is based on Lemma~\ref{lem:post_bel}. After the training finished, $\AstrongGau$ tries to identify the used dataset based on the final posterior beliefs $\beta_k$ on the two datasets. $\AstrongGau$ wins the identification game, if $\AstrongGau$ chooses the used dataset \cali{D}. The advantage to win the experiment is statistically estimated from several identical repetitions of the experiment. $\madiGau$ and \dlt are empirically calculated by counting the cases in which $\beta_k$ for \cali{D} exceeds $0.5$ and $\rb$, respectively. 

One pass over all records in \cali{D} (i.e., one epoch), can comprise multiple update steps. 
In mini-batch gradient descent, a number of $b$ records from \cali{D} is sampled for calculating an update and one epoch results in $|\cali{D}|/b$ update steps. 
In batch gradient descent, all records in \cali{D} are used within one update step, and one epoch consists of a single update step. We operate with batch gradient descent, since it reflects the auxiliary side knowledge of $\Astrong$; thus $k$ denotes the overall number of epochs and training steps. In some of the following experiments we will set $\sens=LS_{\clgrad_i}(\cali{D})$ in Step~\ref{step:gs} by calculating the local sensitivity $LS_{\clgrad_i}$ for the clipped gradients $\clgrad_i$ (cf.~Definition~\ref{def:local_sensitivity}). 
These assumptions are similar to those of white-box MI attacks against federated learning~\cite{Nasr2018}.

The time complexities for calculating dataset sensitivity, posterior belief and advantage are stated in Table~\ref{tab:exp-complexity}. Note that the calculation effort will either lie with $\Astrong$ or the data scientist, depending on whether an audit or an actual attack is performed. The calculation of dataset sensitivity was well parallelizable for the dissimilarity measures considered in this paper.

\begin{table}[t]
      \centering
      \caption{Time complexity for $DS$, $\beta$ and $Adv$}\label{tab:exp-complexity}%
      \label{tab:complexity}
\begin{adjustbox}{width=0.5\linewidth}
\begin{tabular}{|c|c|c|}
\hline
Algorithm & \makecell{Time \\ complexity} & Comment \\ \hline
$DS$ & $O(n^2)$ & \makecell{One-time effort for \\ training dataset.} \\\hline
$\beta$ & $O(nk)$ & \makecell{Computing belief from \\ clipped Batch gradients. } \\\hline
$Adv$ & $O(1)$ & \makecell{Computing $Adv$ for individual \\ training (cf.~\ref{step:output} in Algorithm~\ref{alg:adversary_algo})} \\ \hline
\end{tabular}
\end{adjustbox}
\end{table}
 \begin{algorithm} 
	\caption{$\AstrongGau$ in Deep Learning for Unbounded DP} 
	\label{alg:adversary_algo}
	\begin{algorithmic}[1]
		\Require Neighboring datasets \cali{D},\cali{D'} with $n$,$n'$ records, respectively, $k$, $\theta_0$, $\eta$, $\dpgrad_i$ per training step $i \leq k$, $\cali{M}_i$, $(\eps_i, \dlt_i)$, prior beliefs $\beta_0(\cali{D}) = \beta_0(\cali{D'}) = 0.5$,
		\For{$i\in [k]$} \\
			\textbf{Calculate clipped Batch gradients} \\
			$\clgrad_i(\cali{D})\leftarrow \cali{M}_i(\cali{D}, \sgm=0) $\label{step:dgrad}\\ 
			$\clgrad_i(\cali{D'})\leftarrow \cali{M}_i(\cali{D'}, \sgm=0)$\label{step:d_prime_grad}\\ 
			\textbf{Calculate Sensitivity and $\sgm$} \\
			$\sens \leftarrow GS_{\clgrad}=\clip$ \label{step:gs}\\ 
			$\sgm_i = \sens\sqrt{2\ln(1.25/\dlt_i)}/\eps_i$ \\
			\textbf{Calculate Belief} \\
			$\beta_{i+1}(\cali{D})\leftarrow \frac{\beta_i(\cali{D}) \cdot Pr[\cali{M}_i(\cali{D}, \sgm= \sgm_i)=\dpgrad_i]}{\beta_i(\cali{D})*Pr[\cali{M}_i(\cali{D}, \sgm= \sgm_i)=\dpgrad_i] + \beta_i(\cali{D'}) \cdot Pr[\cali{M}_i(\cali{D'})=\dpgrad_i]}$ \label{algstep:belief} \\
			$\beta_{i+1}(\cali{D}')\leftarrow 1-\beta_{i+1}(\cali{D})$ \label{step:belief_alt}\\
			\textbf{Compute weights} \\
			$\theta_{i+1}\leftarrow \theta_i - \eta \dpgrad_i$ \label{step:update}
		\EndFor \\
		Output \cali{D} if $\beta_k(\cali{D})>\beta_k(\cali{D'})$, \cali{D'} otherwise \label{step:output}
	\end{algorithmic}
\end{algorithm}

\subsection{Method for Auditing \eps}
\label{subsec:auditEps}
In this section we introduce a method to empirically determine the privacy loss \eps. This empirical loss is denoted \feps and is relevant for data scientists. If \feps is close to \eps, the DP perturbation does not add more noise than necessary. However, if \feps is far below \eps, too much noise is added, and utility is unnecessarily lost. We repeat the training process multiple times and use the set of results to calculate $\feps$. The empirical loss \feps can be calculated from different quantities $LS_{\clgrad}$, $\beta_k$, and $\madiGau$ observed during model training: 
\begin{itemize}
    \item From $LS_{\clgrad_1},\ldots,LS_{\clgrad_k}$, the empirical \feps is calculated as follows:
\begin{enumerate*}[label=(\roman*)]
    \item calculate $\sgm_1,\ldots,\sgm_k$ as $\sgm_i = 2\clip / LS_{\clgrad_i} \cdot \sgm$ (cf.~Eq.~\eqref{eq:eps_gauss}) for each repetition of the experiment,
    \item calculate \feps with RDP composition with target \dlt, epochs $k$, and $\vec{\sgm}$ using Tensorflow privacy accountant\footnote{\url{https://github.com/tensorflow/privacy/blob/master/tensorflow_privacy/privacy/analysis/rdp_accountant.py}}, and \item choose the maximum value $\eps'^{\max}$ over all repetitions of the experiment.
\end{enumerate*}

\item From posterior beliefs $\beta$, \feps is calculated by
\begin{enumerate*}[label=(\roman*)]
\item choosing the maximum final posterior belief $\beta_k^{\max}$ for all experiments
    and \item setting $\feps=\beta_k^{\max}/(1-\beta_k^{\max})$ using Eq.~\eqref{eq:epsFromBound}.
\end{enumerate*}

\item From $\madiGau$:
\begin{enumerate*}[label=(\roman*)]
\item counting the number of wins $n_{win}$, i.e., how often $\beta_k>0.5$ over all $n_{Exp}$ experiments,
\item estimate $\madiGau=2 n_{win}/n_{Exp}-1$, and
\item calculate \\ $\feps = \sqrt{2\ln(1.25/\dlt)} \, \Phi^{-1}\left(\frac{\madiGau+1}{2}\right)$ using Eq.~\eqref{eq:adv_to_eps}.
\end{enumerate*}
\end{itemize}
This empirical loss \feps will only be close to \eps if noise is added according to the sensitivity of the dataset. Of the three variants above, the calculation from the sensitivities is the most direct method. The calculation from the posterior belief is less direct. Since the identification advantage ignores the size of the belief it is expected to be the least accurate way to estimate \eps.

Furthermore, we also implement the MI adversary $\Ami$ defined by Yeom et al.~\cite{YGF+18} and compare the resulting advantage to the advantage achieved by $\AstrongGau$. This instance of $\Ami$ uses the loss of a neural network prediction in an approach similar to $\AstrongGau$, who analyzes the gradient updates instead.
\section{Evaluation}
\label{sec:eval}

We empirically show that we can train models that yield an empirical privacy loss \feps close to the specified privacy loss bound $\eps$. We achieve an advantage equal to $\ra$ and tightly bound posterior belief $\rb$ when the sensitivity is set to $LS_{\clgrad_i}$ for the clipped batch gradients at every update step $i$. Privacy is specified by setting the upper bound for the belief, e.g., to~$\rb=0.9$. Together with the sensitivity (cf.~Section~\ref{subsec:settingDPSGD}) this determines the noise of the Gaussian mechanism and yields \eps. The posterior belief $\beta$ and the advantage $\madiGau$ are then empirically determined using the implemented adversary\footnote{We provide code and data for this paper: \url{https://github.com/SAP-samples/security-research-identifiability-in-dpdl}. All experiments within our work were realized by using the Tensorflow privacy package: \url{https://github.com/tensorflow/privacy}.} $\AstrongGau$ as described in Section~\ref{subsec:quantifyIdentifiability}. The empirical privacy loss $\eps'$ is determined as described in Section~\ref{subsec:auditEps}.
We evaluate $\AstrongGau$ for three ML datasets: the MNIST image dataset\footnote{Dataset and detailed description available at: \url{http://yann.lecun.com/exdb/mnist/}}, the Purchase-100 customer preference dataset~\cite{shokri2017}, and the Adult census income dataset~\cite{Koh96}. To improve training speed in our experiments, we set training dataset \cali{D} to a randomly sampled subset of size 100 for MNIST and 1000 for both Purchase-100 and Adult. Multiple trainings and perturbations are evaluated on the sampled \cali{D}. 

The MNIST NN consists of two convolutional layers with kernel size $(3,3)$ each, batch normalization and max pooling with pool size $(2,2)$, and a 10-neuron softmax output layer. For Purchase-100, the NN comprises a 600-neuron input layer, a 128-neuron hidden layer and a 100-neuron output layer. Our NN for Adult consists of a 104-neuron input layer due to the use of dummy variables for categorical attributes, two 6-neuron hidden layers and a 2-neuron output layer. We used relu and softmax activation functions for the hidden layers and the output layer. For all experiments we chose the learning rate $\eta=0.005$ and set the number of iterations $k=30$ which led to converging models. Preprocessing comprised removal of incomplete records, and data normalization.

\subsection{Evaluation of Sensitivities}
\label{sec:exp_sensitivity}

While local sensitivity is favored when striving to reach the privacy bound, we evaluate and compute both $\sens=\hat{LS}_{\clgrad_i}(\cali{D})$ and $\sens=GS_{\clgrad}$, as described in Section~\ref{subsec:settingDPSGD}. In addition, we consider bounded and unbounded DP in our experiments. In order to find the most dissimilar data point for the construction of $\hat{\cali{D}}'$ in Eq.~\eqref{eq:xunbounded} and Eq.~\eqref{eq:xbounded} we require a dissimilarity measure. We considered domain specific candidates for the dissimilarity measures: the negative structural similarity index measure (SSIM) and Euclidean distance for MNIST, and the Hamming, Euclidean, Manhattan, and Cosine distance for the datasets Purchase-100 and Adult. We chose these metrics because we expect them to contain information relevant to the gradients of data points. However, for example we quickly noticed for the Euclidean distance on MNIST image data that it does not capture the meaning or shapes pictured and thus falls short. Instead, the SSIM captures structure in images, and images with a small SSIM dissimilarity values resulted in similar gradients, while images with greater dissimilarity resulted in very different gradients. This observation supports the hypothesis that an appropriate domain-specific measure can be used to estimate local sensitivity $LS_{\clgrad_i}$ from dataset sensitivity $DS$. For Purchases-100 the Hamming distance was clearly superior to the Cosine distance as illustrated in Figures~\ref{fig:sensitivity-boxplot-Purch} and~\ref{fig:sensitivity-boxplot-purch-cosine}. The Manhattan distance worked best for the Adult dataset. For the sensitivity experiments the bound for the posterior belief is set to $\rb=0.9$. Each experiment concerning dataset sensitivity is repeated $n_{Exp}=250$ times.

To confirm that maximizing dataset sensitivity from Definition~\ref{def:dataset_sensitivity} allows us to approximate $LS_{\clgrad_i}$, we train with several differing \cali{D'} and evaluate the sensitivities for all $k=30$ iterations. For the MNIST dataset, the top three choices of \cali{D'} that maximize $DS$ and the three choices that minimize $DS$ are used. As expected, the resulting local sensitivities $LS_{\clgrad_i}$ shown in Figure~\ref{fig:sensitivity-boxplot-MNIST} are clearly larger for the three top choices. The outliers for the second and third smallest dataset sensitivities only account for 1.6\% and 5.2\% of the 7500 overall observed sensitivity norms. More importantly, no far outliers occur for the largest and smallest sensitivities. The same general trend holds for Purchase-100 and Adult in Figures~\ref{fig:sensitivity-boxplot-Purch} and ~\ref{fig:sensitivity-boxplot-Adult}, which we limit to the maximum and minimum $DS$ due to space constraints.

If the chosen global sensitivity is too large compared to the local sensitivity of a specific dataset too much noise will be added when using $GS_{\clgrad}$, as described in Section~\ref{subsec:settingDPSGD}. Global sensitivity $GS_{\clgrad}$ and local sensitivity $LS_{\clgrad_i}$ are determined for bounded and unbounded DP over $n_{Exp}=1000$ repetitions for $\rb=0.9$ ($\eps=2.2$) according to Eq.~\eqref{eq:datapointLSunbounded} and Eq.~\eqref{eq:datapointLS}.
They can be compared in Figure~\ref{fig:local-sensitivity}. 
\begin{figure}
	\centering
    \begin{subfigure}[b]{0.26\linewidth}
        \includegraphics[width=\linewidth]{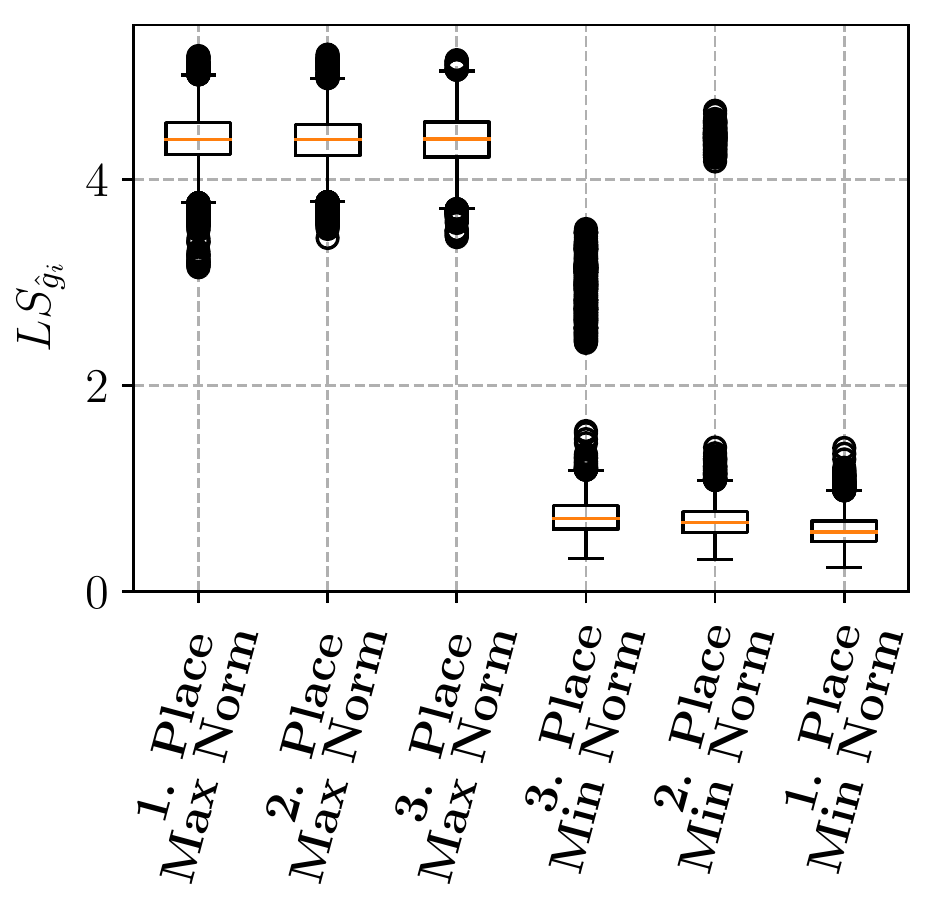} 
        \caption{MNIST: \\ SSIM distance}
        \label{fig:sensitivity-boxplot-MNIST}
    \end{subfigure}
    \begin{subfigure}[b]{0.15\linewidth}
        \includegraphics[width=\linewidth]{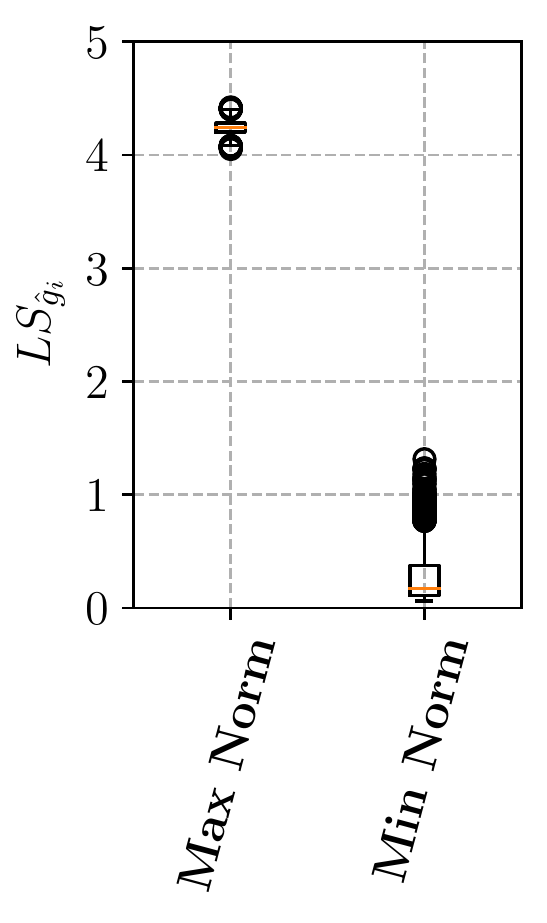}
        \caption{Purchase-100: Hamming distance}
        \label{fig:sensitivity-boxplot-Purch}
    \end{subfigure}
    \begin{subfigure}[b]{0.15\linewidth}  
        \includegraphics[width=\linewidth]{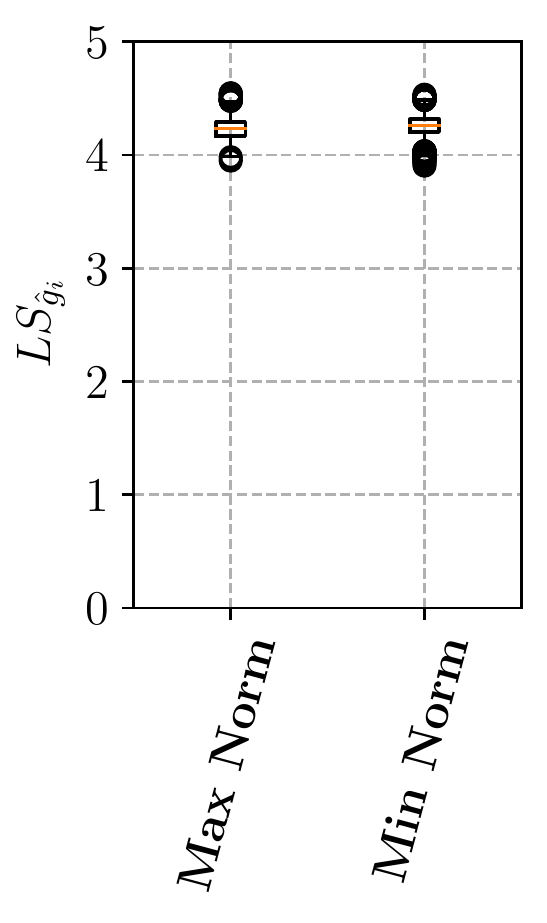}
        \caption{Purchase-100: Cosine distance}
        \label{fig:sensitivity-boxplot-purch-cosine}
    \end{subfigure}
    \begin{subfigure}[b]{0.15\linewidth}
        \includegraphics[width=\linewidth]{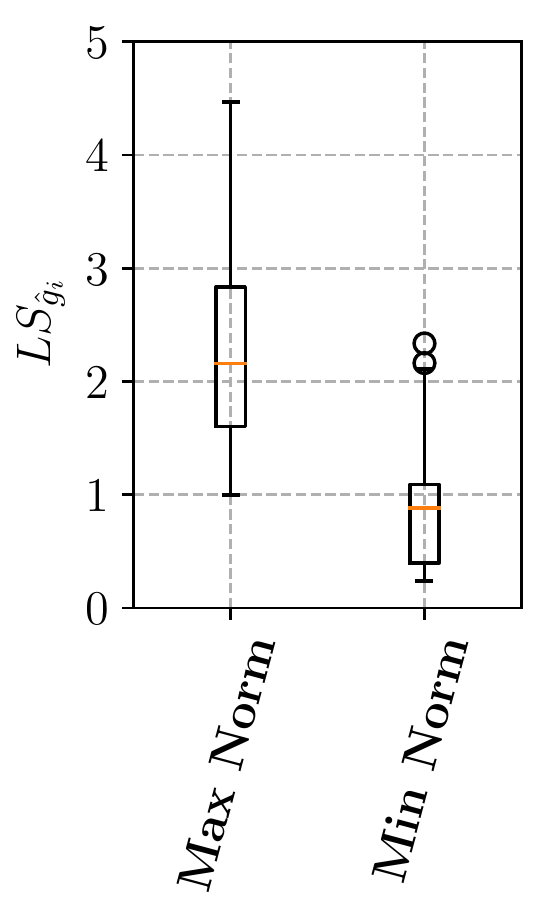}
        \caption{Adult: \\ Manhattan distance}
        \label{fig:sensitivity-boxplot-Adult}
    \end{subfigure}
    \caption{Distribution of the local sensitivity $LS_{\clgrad_i}(\cali{D})$ computed by the adversary using Eq.~\eqref{eq:sensitivityADI} from max to min difference in \cali{D} and \cali{D'} for all $30$ epochs $i$, repeated $250$ times}
    \label{fig:sensitivity-boxplot} 
\end{figure} 

\begin{figure}
	\centering
	\includegraphics[width=0.5\linewidth]{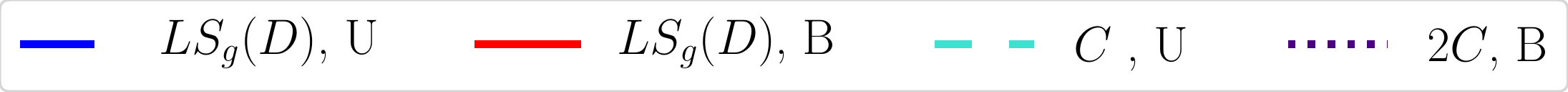}
	\begin{subfigure}[b]{0.225\linewidth}    
	    	\includegraphics[width=\linewidth]{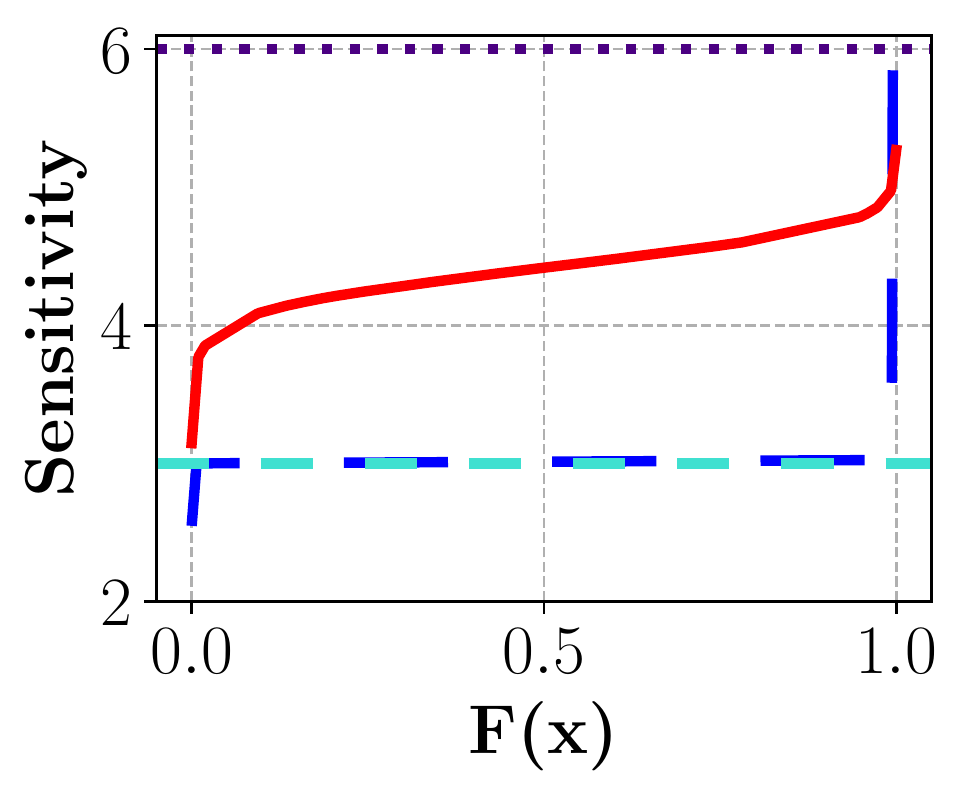}
	    	\subcaption{MNIST}
            \label{fig:local-sensitivity-mnist}
	\end{subfigure}%
    \begin{subfigure}[b]{0.225\linewidth}  
            \includegraphics[width=\linewidth]{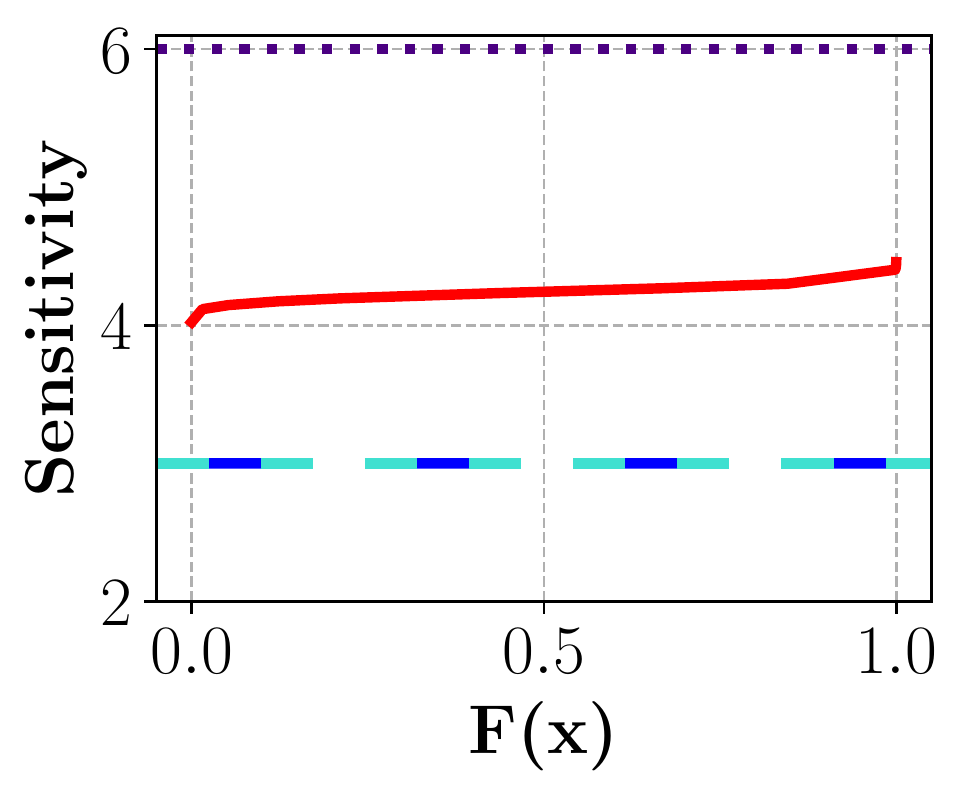}
            \subcaption{Purchase-100}
            \label{fig:local-sensitivity-purch}
    \end{subfigure}%
    \begin{subfigure}[b]{0.225\linewidth}  
            \includegraphics[width=\linewidth]{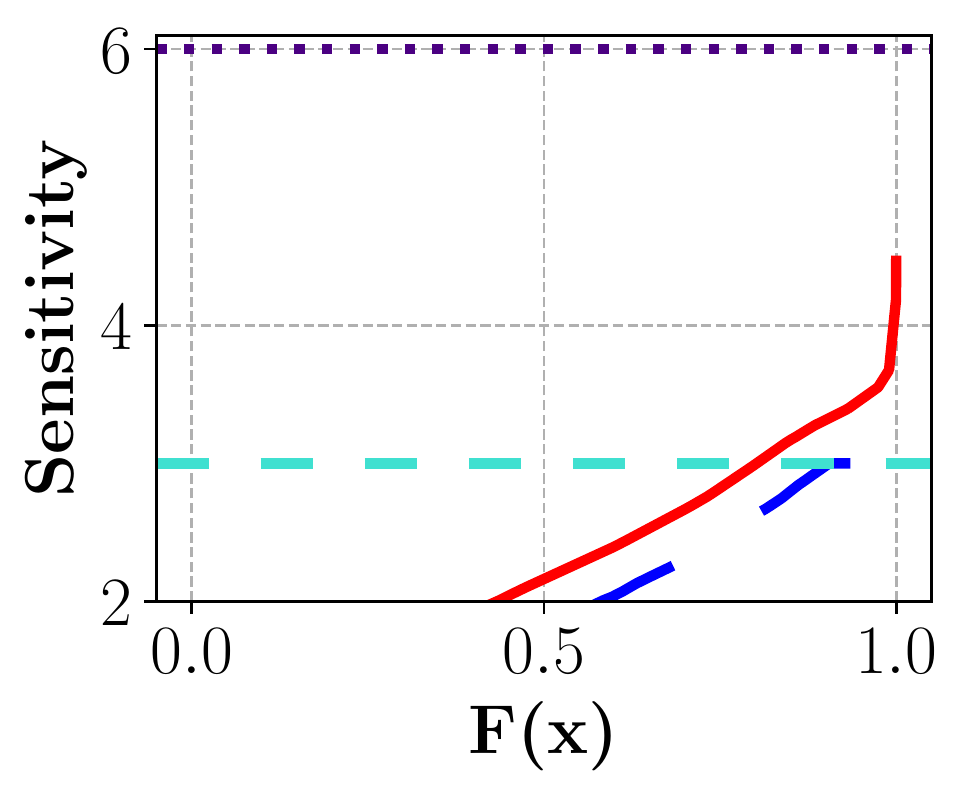}
            \subcaption{Adult}
            \label{fig:local-sensitivity-adult}
    \end{subfigure}  
    \caption{Sensitivities over the course of the training for $\rb=0.9$ ($\eps=2.2$) and $\clip=3$}	
    \label{fig:local-sensitivity}
\end{figure}

\subsection{Quantification of Identifiability for DPSGD}
\label{sec:exp_dpsgd}
For each of the 1000 experiment repetitions, the posterior belief $\beta_k$ and the membership advantage $\madiGau$ are experimentally determined using the implementation of $\AstrongGau$ for DPSGD. We set $\rb=0.9$ ($\eps=2.2$) and compare bounded and unbounded DP. Table~\ref{tab:emp-success} shows the analytically obtained values for privacy loss \eps, and the bound $\ra$ for the advantage. The parameters \eps, \dlt, and $\ra$ for $\rb=0.9$ can be read from Table~\ref{tab:exp-parameters}; \eps is determined from Eq.~\eqref{eq:epsFromBound}, whereas $\ra$ is calculated from \eps from Theorem~\ref{def:expected-adversarial-success-rate}.

First, we verify that the upper bound $\rb$ on the posterior belief holds. The posterior beliefs $\beta_k$ of these experiments are described in Figures~\ref{fig:belief-boxplot-mnist}, \ref{fig:belief-boxplot-purch} and \ref{fig:belief-boxplot-Adult}. For a single experiment the posterior belief on the training dataset $\cali{D}$ is on average only slightly above 0.5. While for most cases the posterior belief is far below the bound of 0.9 (specified by the blue, dashed line), the upper bound is violated with a small probability. The relative frequency of these violations is denoted as $\dlt'$. Since the DP bound, and thus $\rb$, only holds with probability $1-\dlt$ according to Theorem~\ref{cor:general-upper-bound} violations are acceptable as long as $\dlt' \le \dlt$. Indeed, the experimentally obtained $\dlt'$ for $\rb=0.9$ in Table~\ref{tab:emp-success} is always smaller than the corresponding $\dlt$ in Table~\ref{tab:exp-parameters}. Similarly, the advantage should be close to the estimate $\ra$ stated in Table~\ref{tab:exp-parameters}. The advantage is experimentally estimated as the relative frequency of experiments where the implemented adversary $\AstrongGau$ correctly chooses $\cali{D}$ and is stated in Table~\ref{tab:emp-success}.

\begin{figure}
	\centering
	\begin{subfigure}[b]{0.25\linewidth}    
	    	\includegraphics[width=\linewidth]{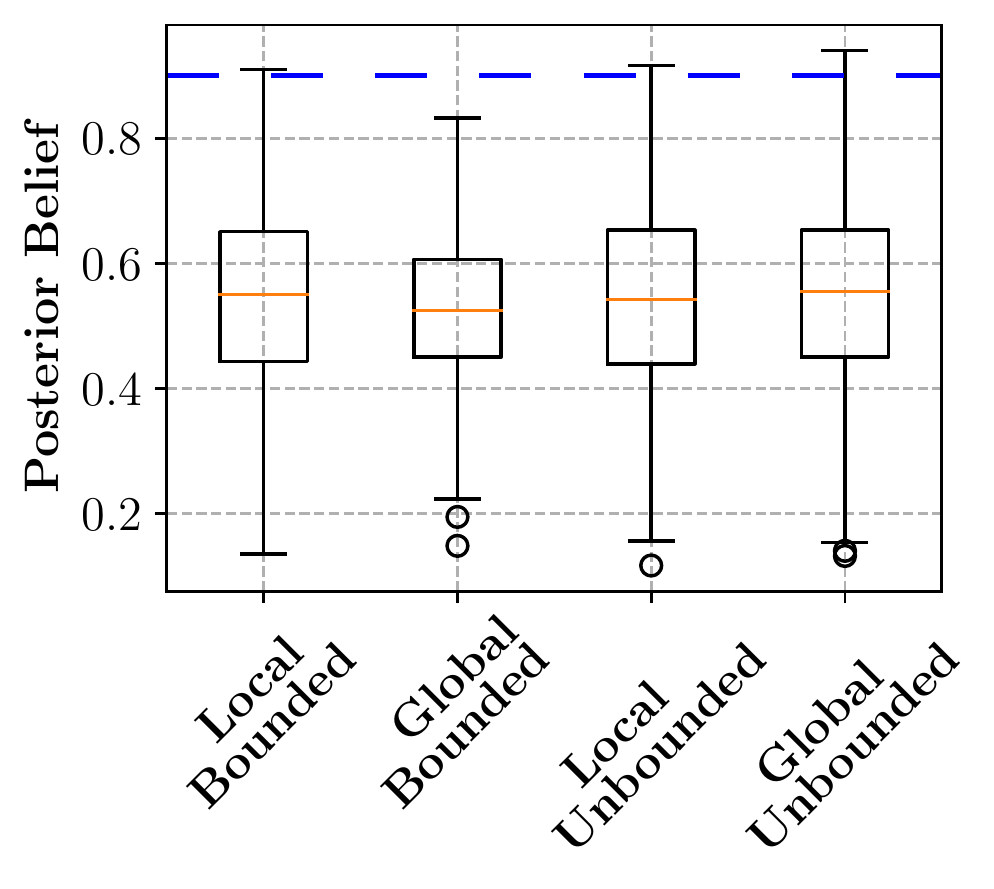}
	        \subcaption{MNIST}
	        \label{fig:belief-boxplot-mnist}
	\end{subfigure}%
    \begin{subfigure}[b]{0.25\linewidth}  
            \includegraphics[width=\linewidth]{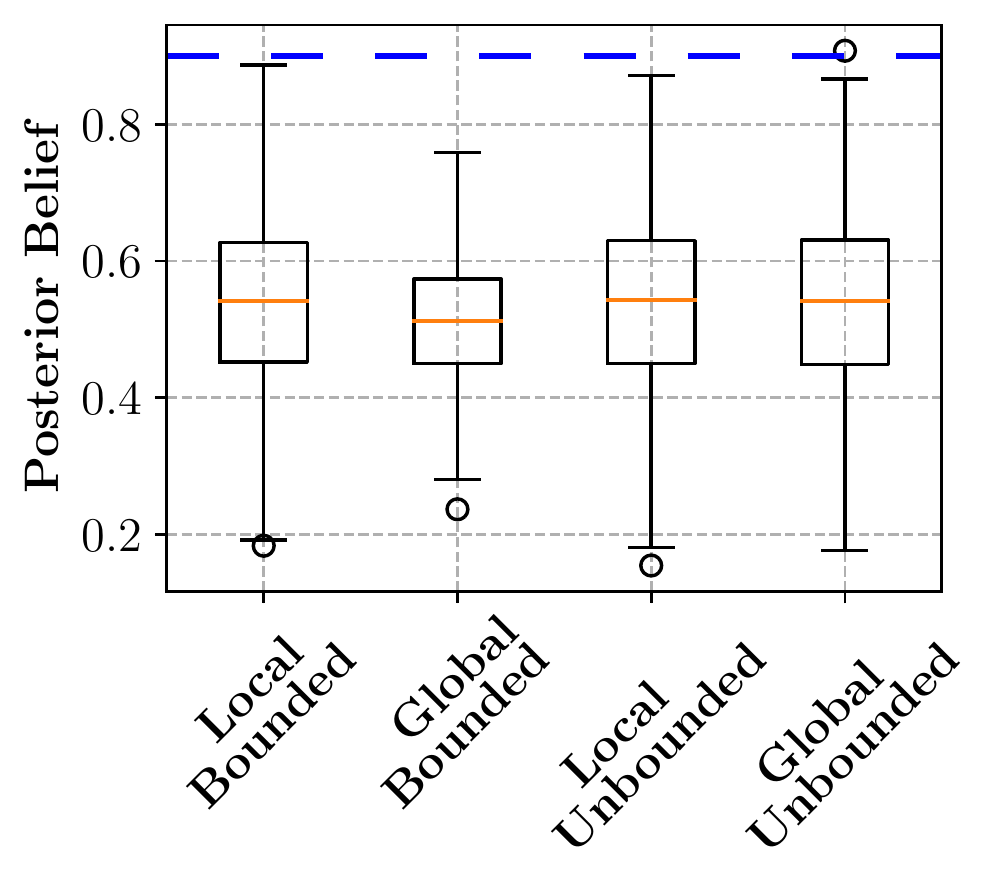}
            \subcaption{Purchase-100}
            \label{fig:belief-boxplot-purch}
    \end{subfigure}%
    \begin{subfigure}[b]{0.25\linewidth}  
            \includegraphics[width=\linewidth]{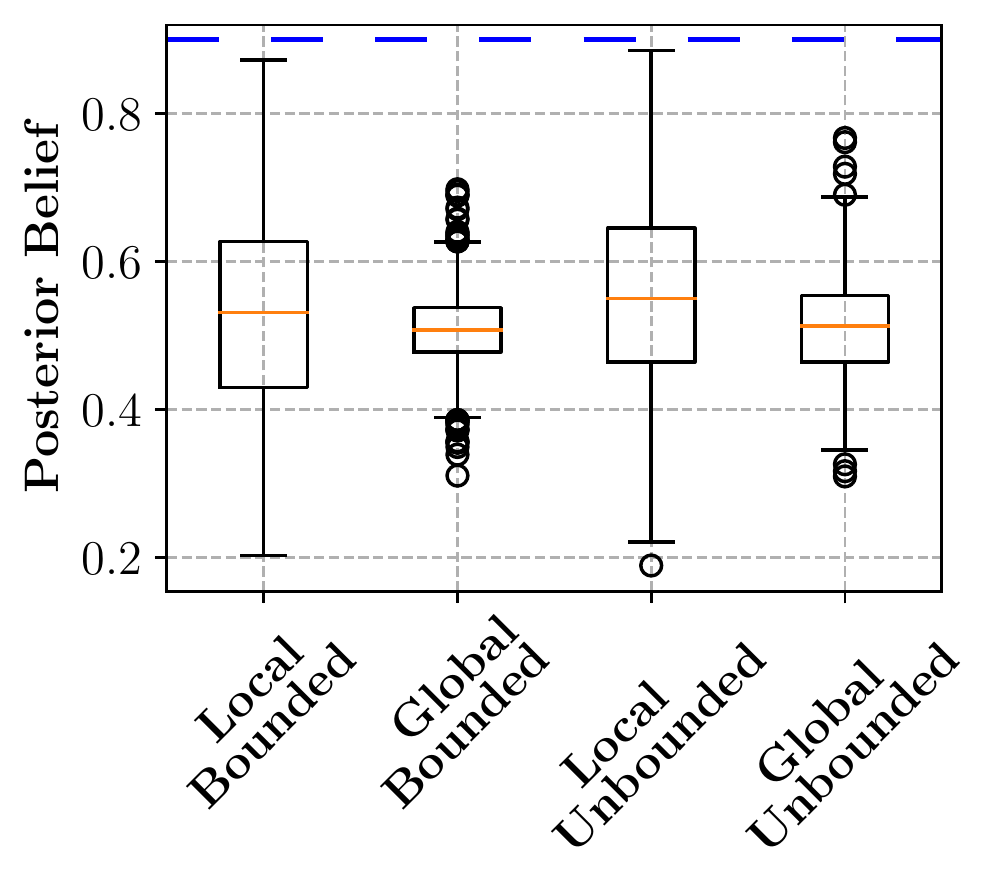}
            \subcaption{Adult}
            \label{fig:belief-boxplot-Adult}
    \end{subfigure}%
    \begin{subfigure}[b]{0.25\linewidth}
        \includegraphics[width=\linewidth]{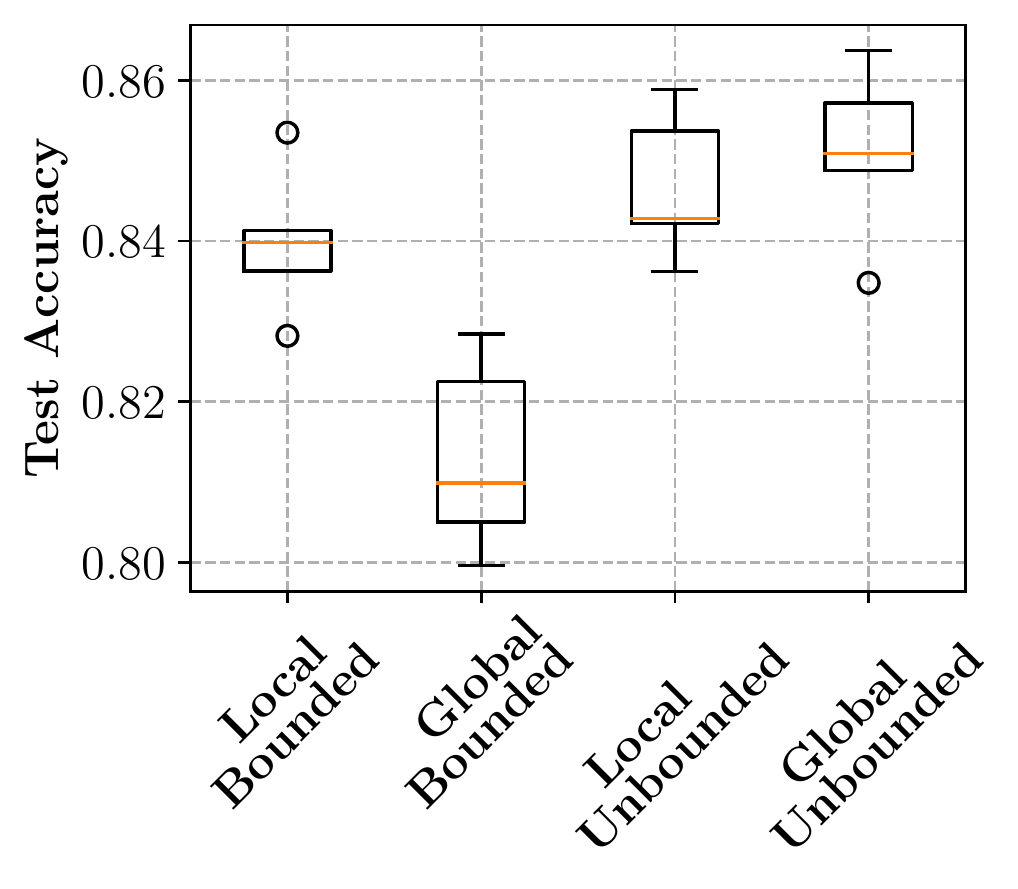}
        \caption{MNIST test accuracy}	
        \label{fig:accuracy-boxplot}
    \end{subfigure}%
    \caption{Distribution of empirical posterior beliefs $\beta_k$ (panels a to c) and an example for test accuracy after training with $\rb=0.9$ ($\eps=2.2$), for local and global sensitivity, bounded and unbounded DP (panel d)}	
    \label{fig:belief-boxplot}
\end{figure}

\begin{table*}[t]
  \begin{minipage}[t]{\linewidth}
      \centering
      \caption{Overview of experiments }%
      
      \subfloat[\label{tab:emp-success} Empirical $\madiGau$ and $\dlt'$ for $\rb=0.9$ using \\ $LS_{\clgrad_i}$ and $GS_{\clgrad}$ with bounded (B) and unbounded (U) DP]{
\resizebox{0.45\columnwidth}{!}{
\begin{tabular}{|cc|c|c|c|c|c|c|}
\hline
                                          &   & \multicolumn{2}{c|}{MNIST}       & \multicolumn{2}{c|}{Purchase-100} & \multicolumn{2}{c|}{Adult}       \\ \cline{3-8} 
                                          &   & $\madiGau$ & $\dlt'$ & $\madiGau$  & $\dlt'$ & $\madiGau$ & $\dlt'$ \\ \hline
\multicolumn{1}{|c|}{\multirow{2}{*}{LS}} & B & 0.24       & 0.002               & 0.25        & 0                   & 0.17       & 0                   \\ \cline{2-8} 
\multicolumn{1}{|c|}{}                    & U & 0.23       & 0.002               & 0.23        & 0                   & 0.22       & 0                   \\ \hline
\multicolumn{1}{|c|}{\multirow{2}{*}{GS}} & B & 0.18       & 0                   & 0.1         & 0                   & 0.13       & 0                   \\ \cline{2-8} 
\multicolumn{1}{|c|}{}                    & U & 0.27       & 0.004               & 0.24        & 0.001               & 0.18       & 0              \\ \hline
\end{tabular}
}}
\subfloat[\label{tab:exp-parameters}Experiment setting for posterior belief $\rb$ and $\dlt$ with analytically determined privacy loss $\eps$ and the advantage bound $\ra$]{
\resizebox{0.55\columnwidth}{!}{
\begin{tabular}{|c|c|c|c|c|c|c|c|c|c|c|c|c|}
\hline
 & \multicolumn{4}{c|}{MNIST} & \multicolumn{4}{c|}{Purchase-100} & \multicolumn{4}{c|}{Adult} \\ \hline
$\rb$ & 0.52 & 0.75 & 0.9 & 0.99 & 0.53 & 0.75 & 0.9 & 0.99 & 0.53 & 0.75 & 0.9 & 0.99 \\ \hline
$\dlt$ & \multicolumn{4}{c|}{0.01} & \multicolumn{4}{c|}{0.001} & \multicolumn{4}{c|}{0.001} \\ \hline
$\eps$ & 0.08 & 1.1 & 2.2 & 4.6 & 0.12 & 1.1 & 2.2 & 4.6 & 0.12 & 1.1 & 2.2 & 4.6 \\ \hline
$\ra$ & 0.01 & 0.14 & 0.28 & 0.54 & 0.01 & 0.12 & 0.23 & 0.46 & 0.01 & 0.12 & 0.23 & 0.46 \\ \hline
\end{tabular}	}
    }

  \end{minipage}\enspace\enspace%
\end{table*}

Figure~\ref{fig:belief-boxplot} illustrates the influence of sensitivity in the bounded and unbounded DP settings. In Figures~\ref{fig:belief-boxplot-mnist}, \ref{fig:belief-boxplot-purch} and \ref{fig:belief-boxplot-Adult}, the chosen upper bound $\rb=0.9$ (blue line) is clearly not reached for the bounded case when global sensitivities are used. Similarly, the advantage of $\AstrongGau$ in Table~\ref{tab:emp-success} is smaller when the global sensitivity is used. Here it holds that $LS_{\clgrad_i}(\cali{D}) < 2\clip=\sens$, which implies that the examples differing between \cali{D'} and \cali{D} do not point in opposite directions in the bounded setting. For the unbounded DP case, this effect is not observed with the MNIST and Purchase-100 datasets. Instead, the use of local and global sensitivity leads to the same distribution of posterior beliefs and approximately the same advantage. This result stems from the fact that the per-example gradients over the course of all epochs were close to or greater than $\clip=3$, i.e., the differentiating example in \cali{D} must have the gradient magnitude $\clip=3$. However, in the Adult dataset, $LS_{\clgrad_i}(\cali{D}) < \clip=3$, so too much noise is added using $GS_{\clgrad}$ in the unbounded DP setting as well. 

From a practical standpoint, these observations are critical, since unnecessary noise degrades the utility of the model when the global sensitivity is too large, as shown in Figure~\ref{fig:accuracy-boxplot}. While all experiments were done with $\clip=3$, we expect a similar relationship between $LS_{\clgrad_i}$ and $GS_{\clgrad}$ for different values of $\clip$, since we observed the unclipped gradients to usually be greater than $\clip=3$.

\subsection{Auditing DPSGD}
\label{sec:exp_audit}

This section details the audit of $\eps$. As shown in Section~\ref{subsec:auditEps}, the calculation of the empirical loss $\feps$ can be based on (i) the local sensitivity, (ii) the posterior beliefs $\beta_k$ or (iii) on the advantage $\madiGau$. To validate that the empirical loss $\feps$ is close to the target privacy loss $\eps$ we use the setting described in Section~\ref{sec:exp_dpsgd} and Table~\ref{tab:exp-parameters}. 
 
The resulting empirical loss $\feps$ is compared to the the target privacy loss $\eps$ for the bounded case in Figures~\ref{fig:mnist_audit} to~\ref{fig:adult_audit}. As expected Figures~\ref{fig:mnist_sens}, \ref{fig:purch_sens} and \ref{fig:adult_sens} support that the privacy loss $\eps$ can be best estimated from the local sensitivity: the red curve lies on the ideal green curve. The estimation is less precise from the posterior beliefs and shown in Figures~\ref{fig:mnist_beta}, \ref{fig:purch_beta} and \ref{fig:adult_beta}. The estimation is worst from the advantage in Figures~\ref{fig:mnist_madi}, \ref{fig:purch_madi} and \ref{fig:adult_madi}, where the red curve deviates most from the ideal green curve for all datasets. It is evident that the use of global sensitivity (blue lines) results in an underestimation of $\eps$ for all datasets. When local sensitivity is used, the small deviation from the ideal curve confirms that $\AstrongGau$ comes close to the theoretical privacy guarantees offered by DP. A data scientist who specifies \eps via the identifiability bounds $\ra$ and $\rb$ can audit \eps using the implementation of $\AstrongGau$. We see that in some cases $\feps > \eps$, or equivalently $\beta_k{(\cali{D})}>\rb$. These variations are due to the probabilistic nature of the estimation and the bound only holds with probability 1-\dlt. Furthermore, we observe in some occasions that $\madiGau > \ra$ which stems from the fact that $\madiGau$ is an expected value for a series of experiments, which falls within a confidence interval around $\ra$.

\begin{figure*}[ht!]
    \centering\includegraphics[width=0.5\textwidth]{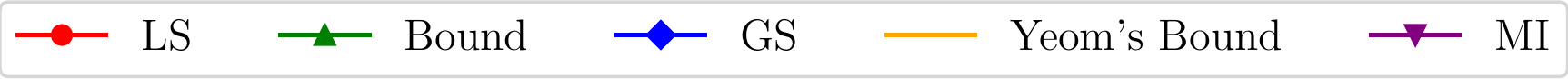} \\
    \begin{subfigure}[b]{0.225\linewidth}  
            \includegraphics[width=\linewidth]{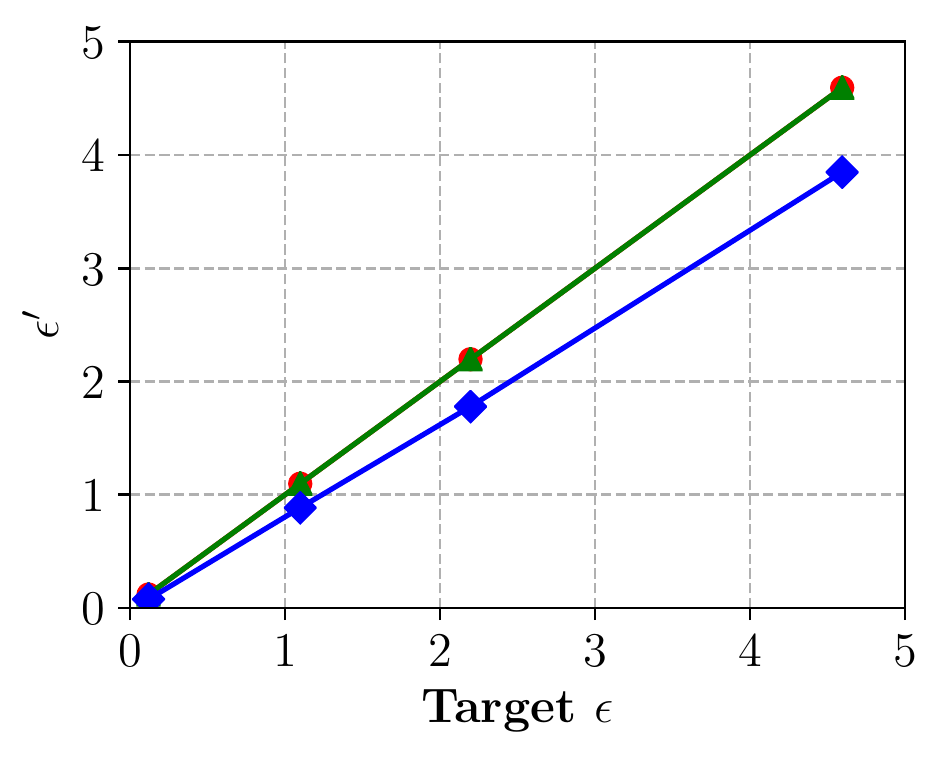}
            \subcaption{\feps from $\sens_{0},\ldots,\sens_{k}$}
            \label{fig:mnist_sens}
    \end{subfigure}
	\begin{subfigure}[b]{0.26\linewidth}    
	    	\includegraphics[width=\linewidth]{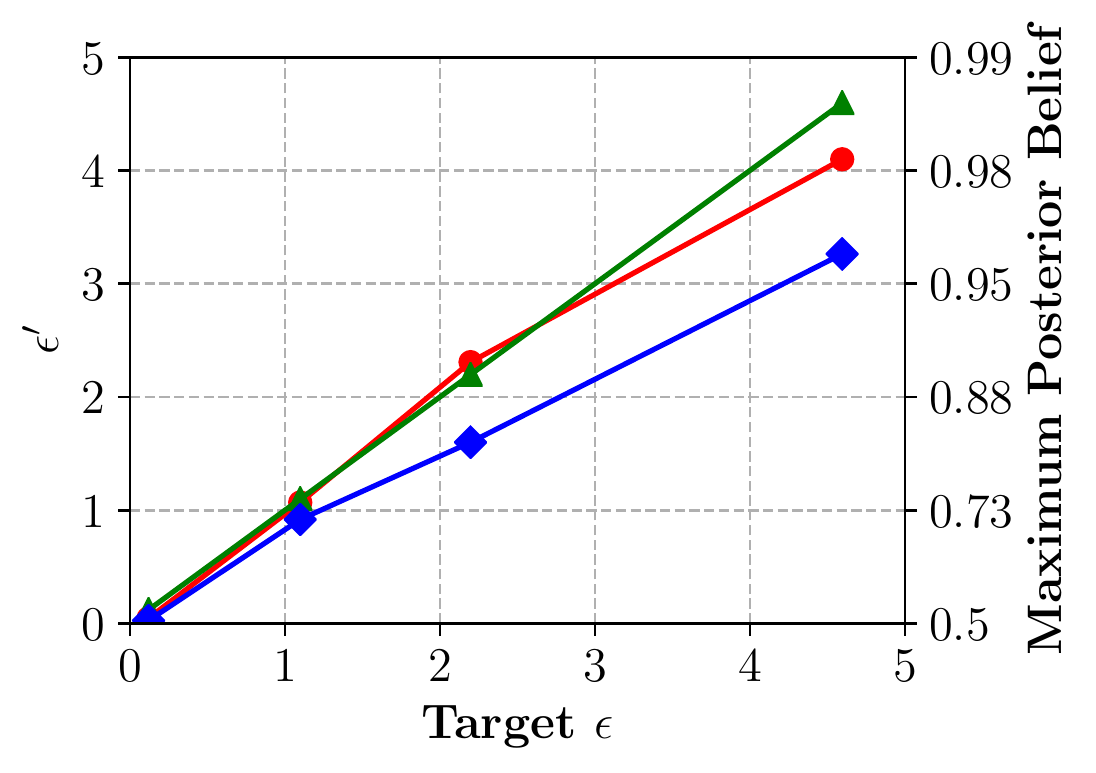}
	        \subcaption{\feps from posterior belief $\beta_k$}
            \label{fig:mnist_beta}
	\end{subfigure}%
    \begin{subfigure}[b]{0.26\linewidth}  
            \includegraphics[width=\linewidth]{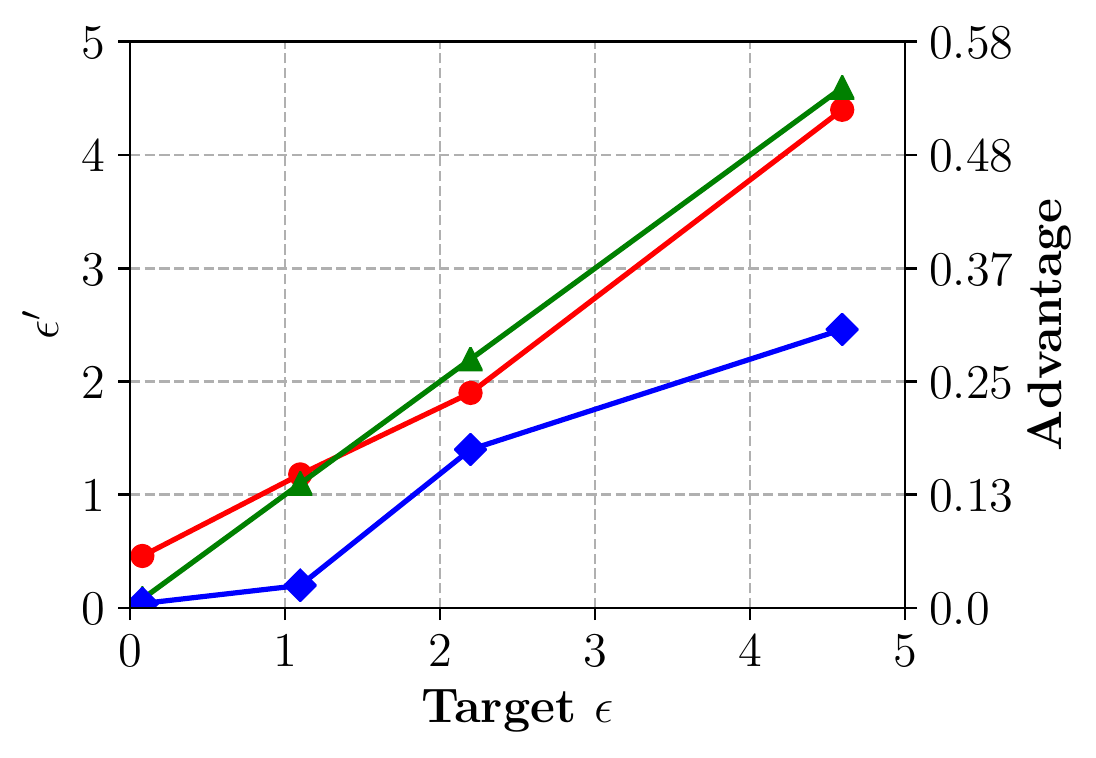}
            \subcaption{\feps from advantage $\madiGau$}
            \label{fig:mnist_madi}
    \end{subfigure}
    \begin{subfigure}[b]{0.23\linewidth}  
            \includegraphics[width=\linewidth]{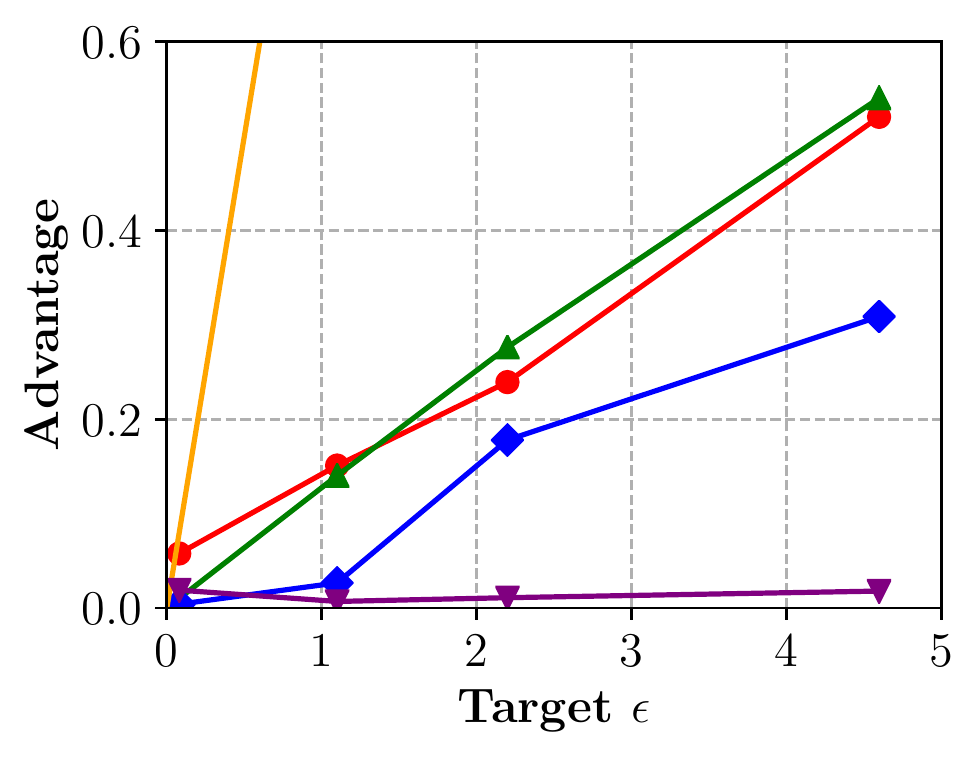}
            \subcaption{Comparison to MI}
            \label{fig:mnist_mi}
    \end{subfigure}
    \caption{Audit of $\eps$ (a-c) and comparison with $\Ami$ (d) for MNIST data (bounded case)}
    \label{fig:mnist_audit}
\end{figure*}%
\begin{figure*}[ht!]
    \centering
    \begin{subfigure}[b]{0.225\linewidth}  
            \includegraphics[width=\linewidth]{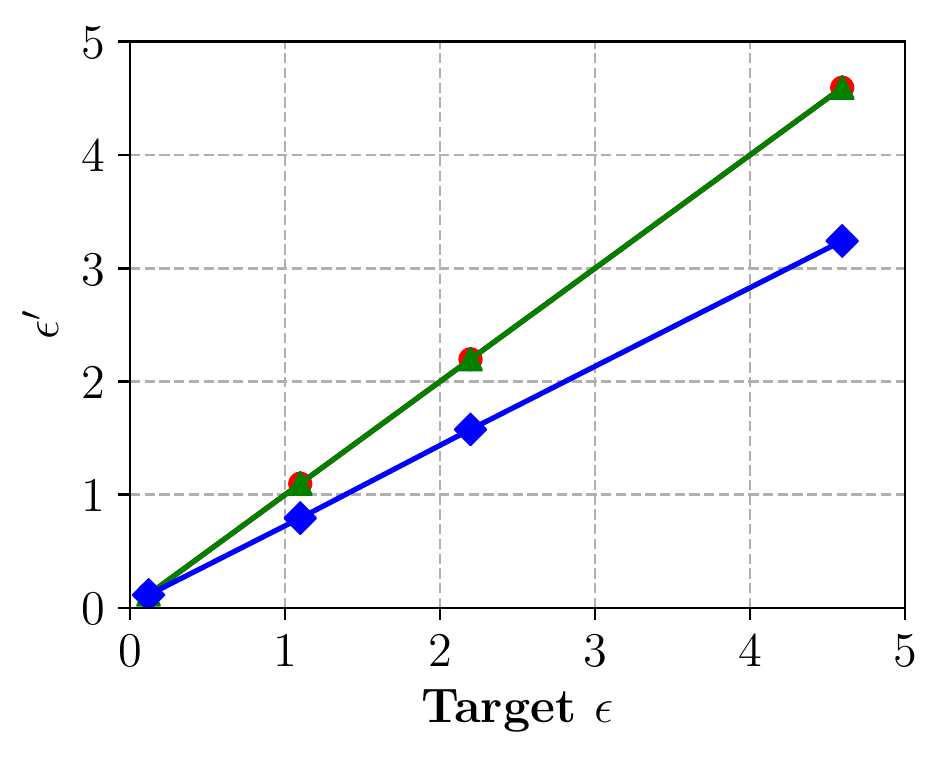}
            \subcaption{\feps from $\sens_{0},\ldots,\sens_{k}$}
            \label{fig:purch_sens}
    \end{subfigure}
	\begin{subfigure}[b]{0.26\linewidth}    
	    	\includegraphics[width=\linewidth]{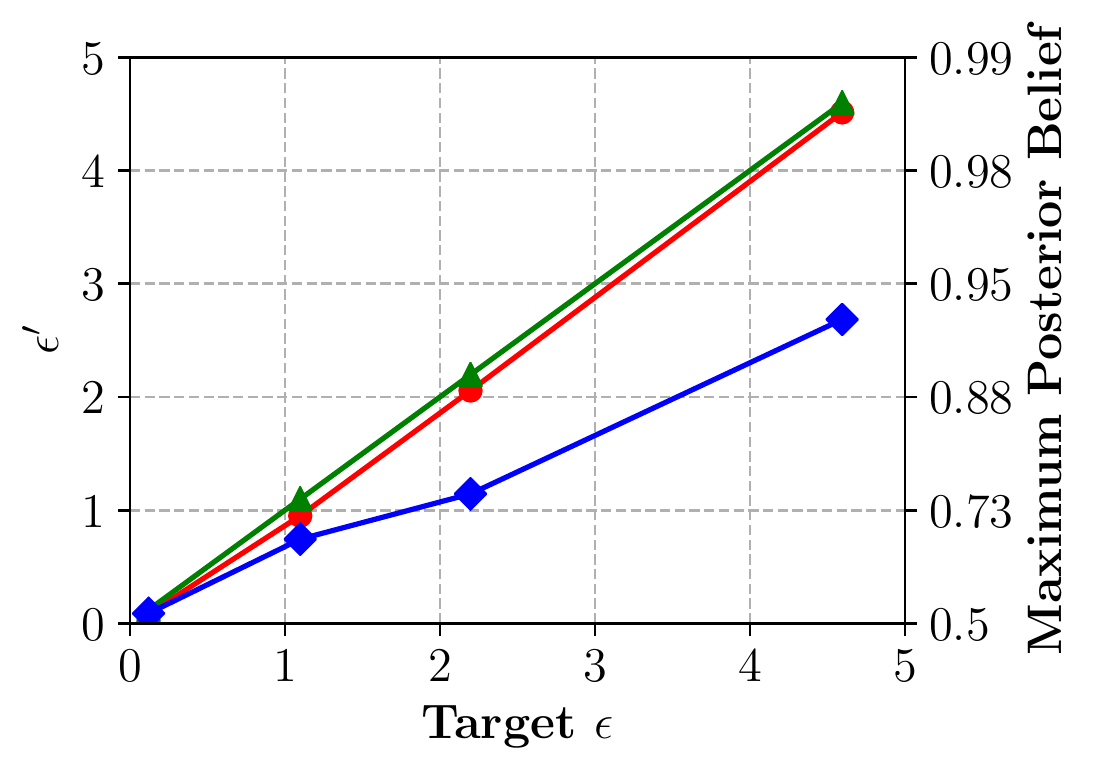}
	        \subcaption{\feps from posterior belief $\beta_k$}
            \label{fig:purch_beta}
	\end{subfigure}%
    \begin{subfigure}[b]{0.26\linewidth}  
            \includegraphics[width=\linewidth]{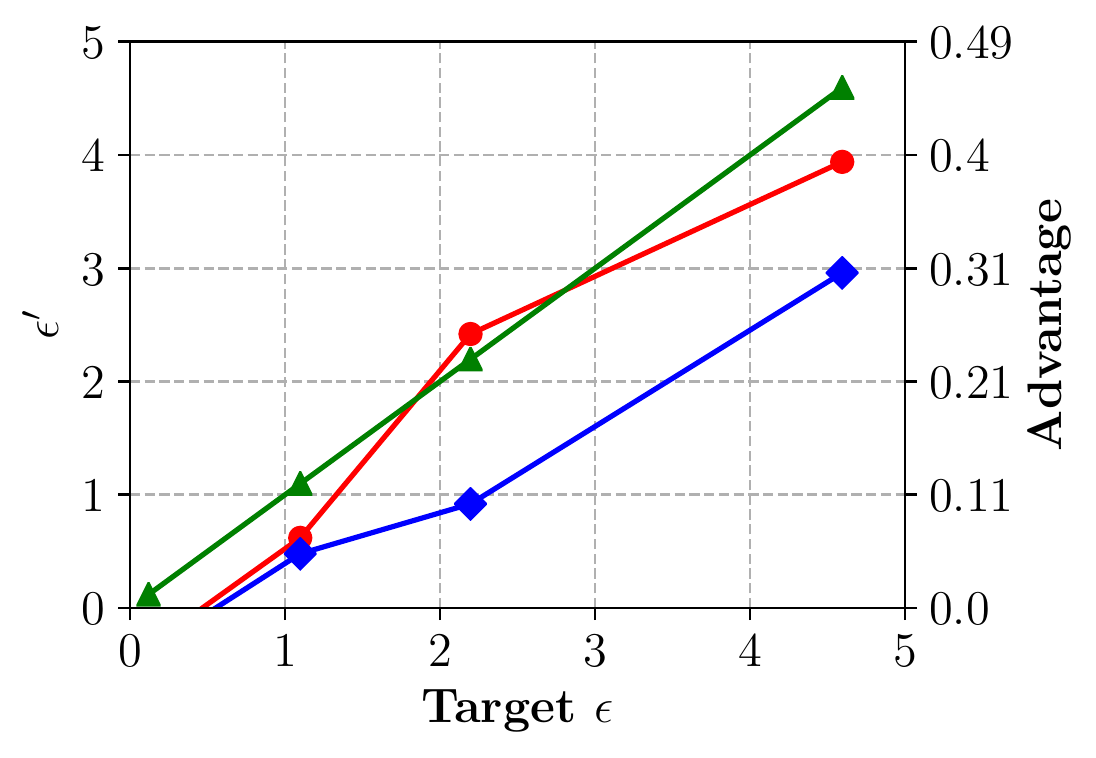}
            \subcaption{\feps from advantage $\madiGau$}
            \label{fig:purch_madi}
    \end{subfigure}
    \begin{subfigure}[b]{0.225\linewidth}  
            \includegraphics[width=\linewidth]{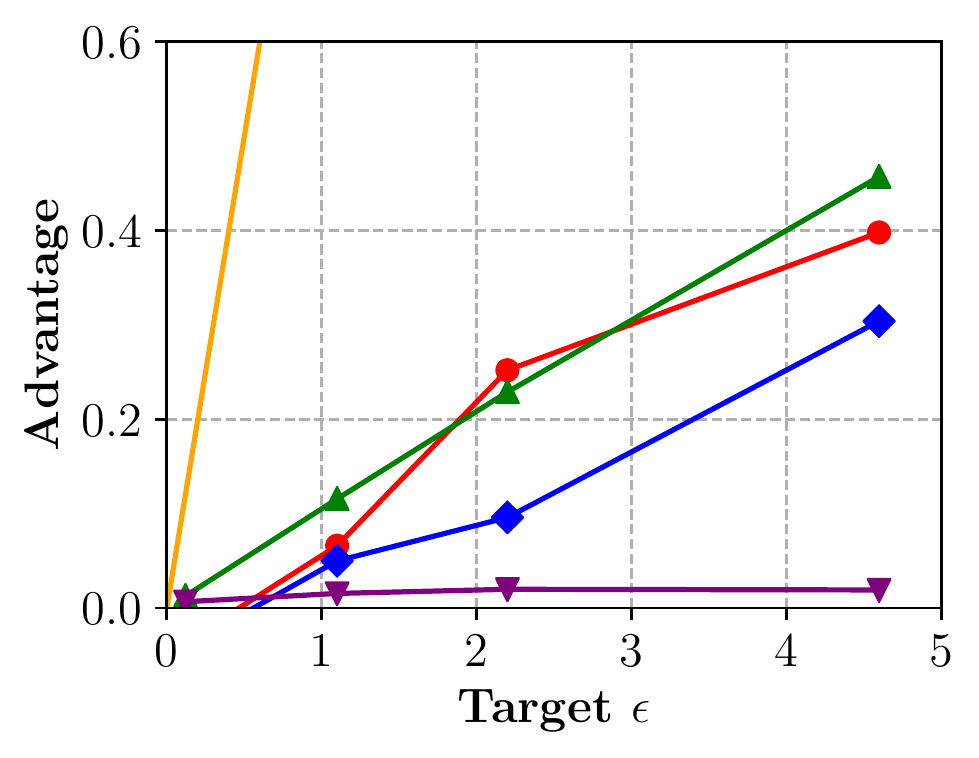}
            \subcaption{Comparison to MI}
            \label{fig:purchase_mi}
    \end{subfigure}
    \caption{Audit of $\eps$ (a-c) and comparison with $\Ami$ (d) for Purchase-100 data (bounded case)}
    \label{fig:purchase_audit}
\end{figure*}%
\begin{figure*}[h!]
    \centering
    \begin{subfigure}[b]{0.225\linewidth}  
            \includegraphics[width=\linewidth]{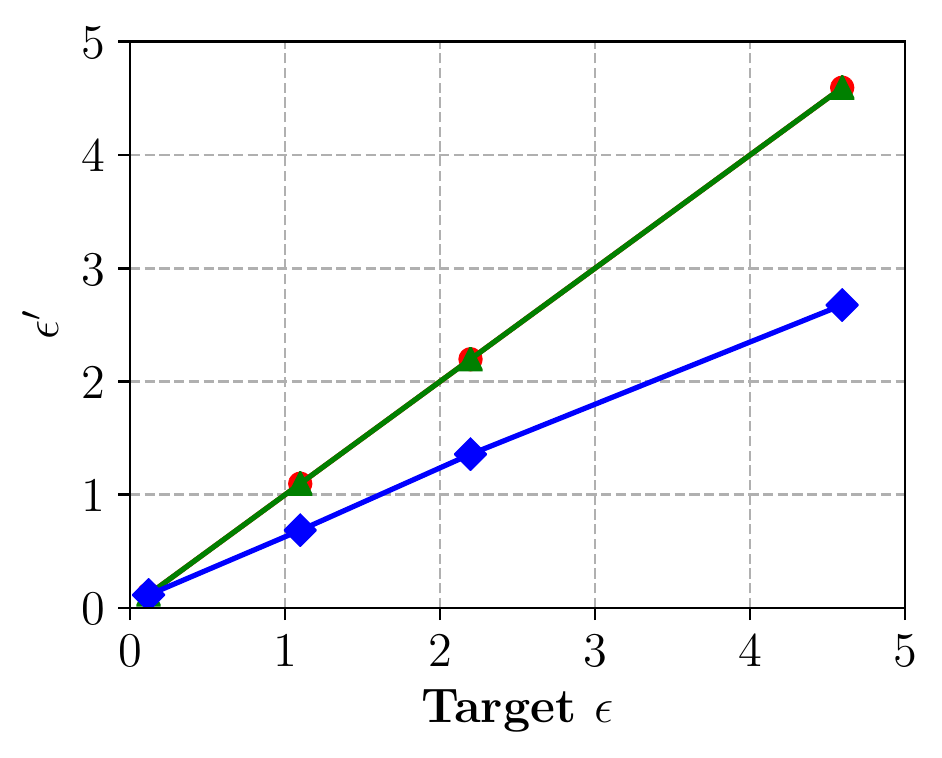}
            \subcaption{\feps from $\sens_{0},\ldots,\sens_{k}$}
            \label{fig:adult_sens}
    \end{subfigure}
	\begin{subfigure}[b]{0.26\linewidth}    
	    	\includegraphics[width=\linewidth]{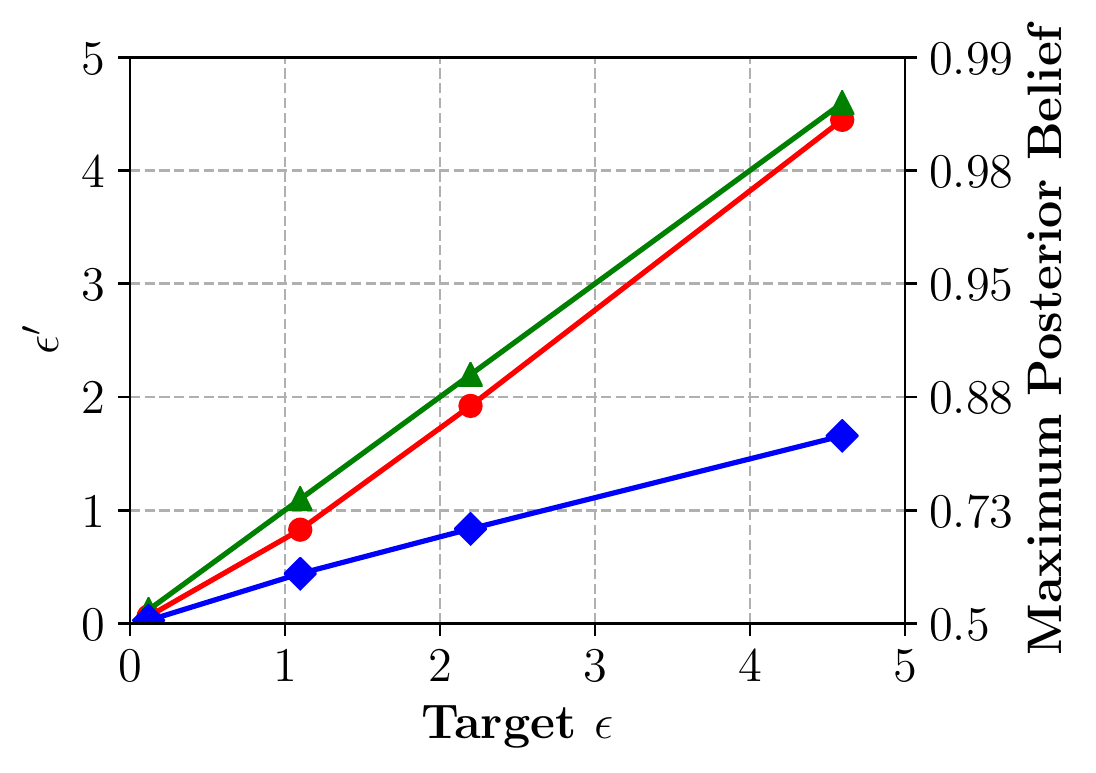}
	        \subcaption{\feps from posterior belief $\beta_k$}
            \label{fig:adult_beta}
	\end{subfigure}%
    \begin{subfigure}[b]{0.26\linewidth}  
            \includegraphics[width=\linewidth]{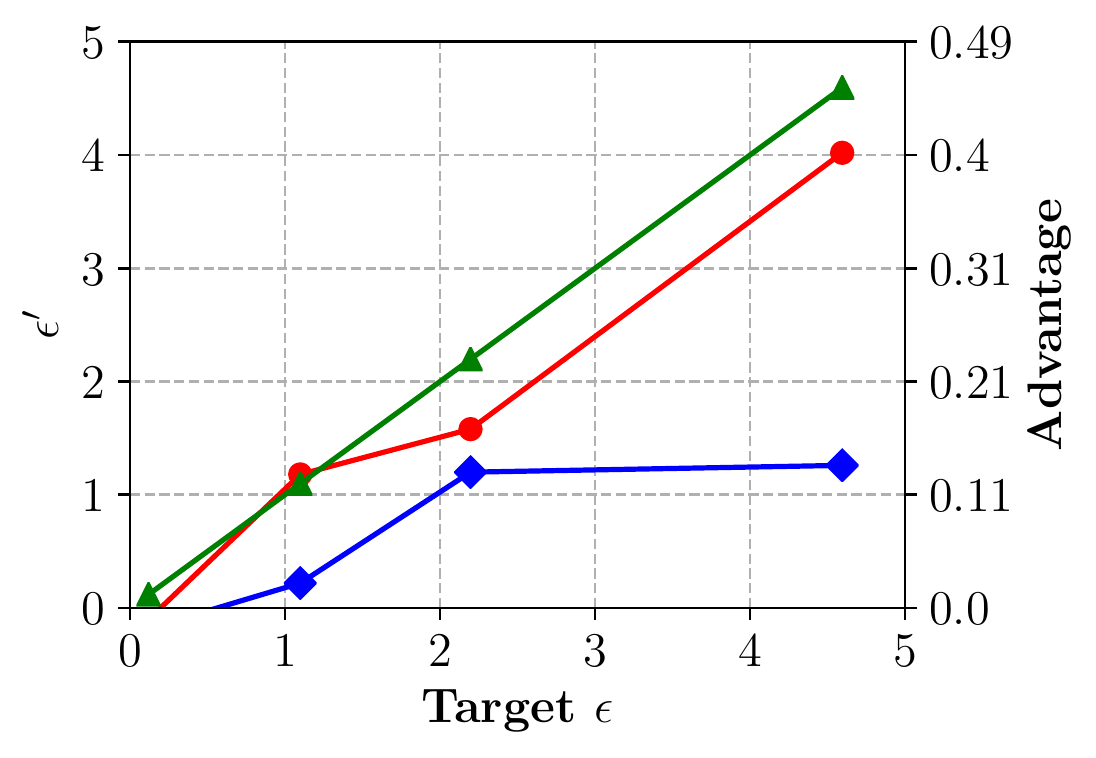}
            \subcaption{\feps from advantage $\madiGau$}
            \label{fig:adult_madi}
    \end{subfigure}
    \begin{subfigure}[b]{0.23\linewidth}  
            \includegraphics[width=\linewidth]{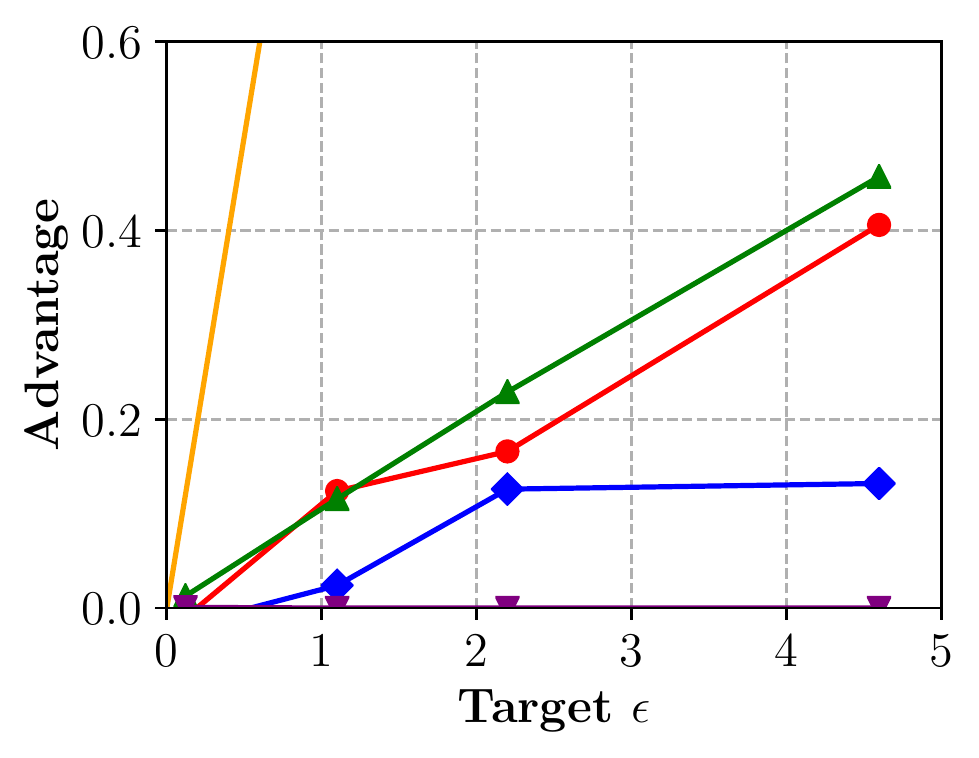}
            \subcaption{Comparison to MI}
            \label{fig:adult_mi}
    \end{subfigure}
    \caption{Audit of $\eps$ (a-c) and comparison with $\Ami$ (d) for Adult data (bounded case)}
    \label{fig:adult_audit}
\end{figure*}

To enable comparison with membership inference we implemented $\Ami$ by expanding the implementation of Jayaraman and Evans~\cite{JE19}, which implements the attack suggested by Yeom et al.~\cite{YGF+18}. Figures~\ref{fig:mnist_mi},~\ref{fig:purchase_mi}, and~\ref{fig:adult_mi} visualize the advantage resulting from both $\AstrongGau$ and $\Ami$ for our setting, as well as the bounds provided by the DP guarantee and the MI bound of Yeom et al.~\cite{YGF+18}. We see that the MI bound is very loose for all evaluated datasets, as previously noted by Jayaraman and Evans~\cite{JE19}. Furthermore, we see that our implementation of $\AstrongGau$ significantly outperforms $\Ami$ on all datasets and values of $\eps$.

\section{Discussion}
\label{sec:disc}
 $\Astrong$ diverges from other attacks against DP or ML, which necessitates a discussion of $\Astrong$'s properties in relation to alternative approaches. Our goal is to construct an adversary that most closely challenges DP, and can be conntected to societal norms and legislation via identifiability score. To this end, $\Astrong$ has knowledge of all but one element of the training data and the gradients at every update step. Since the DP guarantee must hold in the presence of all auxiliary information, both of these assumptions relate the attack model $\Astrong$ directly to the DP guarantee. Since $\Astrong$ has knowledge of all but one element instead of only the distribution, $\Astrong$ possesses significantly more information than MI adversaries. 

A natural question arises w.r.t.~$\Astrong$'s practical relevance. Especially in a federated learning setting $\Astrong$ knows the gradients during every update step, if participating as a data owner. Furthermore, $\Astrong$ could realistically obtain knowledge of a significant portion of the training data, since public reference data is often used in training datasets and only extended with some custom training data records, necessitating the notion of DP in general. 

To further comment on the utility that can be achieved from a differentially private model, we note that the optimal choice for \clip may stray from the original recommendation of Abadi et al.~\cite{abadi2016}. We follow this recommendation and set $\clip=3$, which limits the utility loss that results when \clip is too large (unnecessary noise addition) and too small (loss of information about the gradient). 
Since this balance holds for unbounded DP and does not consider the notion of local sensitivity, we expect that a different \clip may yield better utility than what we report. Varying \clip may also change the balance between local sensitivity and global sensitivity from Figures~\ref{fig:mnist_audit} to~\ref{fig:adult_audit}. 
Furthermore, since gradients change over the course of training, the optimal
value of \clip at the beginning of training may no longer be optimal
toward the end of training according to McMahan et al.~\cite{McMahan2018b}. Adaptively setting the clipping norm as suggested by Thakkar et al.~\cite{Thakkar} may improve utility by changing \clip as training progresses. We expect that doing so might bring \feps closer to \eps when auditing the DP guarantee, and achieve similar by using local sensitivity.

\section{Related Work}
\label{sec:rel}
Choosing and interpreting DP privacy parameters has been addressed from several directions.

Lee and Clifton~\cite{Lee2011, Lee2012} proposed DI as a Bayesian privacy notion which quantifies \eps w.r.t.~an adversary's maximum posterior belief $\rb$ on a finite set of possible input datasets. Yet, both papers focus on the scalar \eps Laplace mechanism without composition, while we consider the \epsdlt multidimensional Gauss mechanism under RDP composition. Li et al.~\cite{Li13} demonstrate that DI matches the DP definition when an adversary decides between two neighboring datasets $\cali{D}, \cali{D'}$. Kasiviswanathan et al.~\cite{KS14} also provide a Bayesian interpretation of DP. While they also formulate posterior belief bounds and discuss local sensitivity, they do not cover expected advantage and implementation aspects such as dataset sensitivity. 

The choice of privacy parameter \eps has been tied to economic consequences. Hsu et al.~\cite{Hsu2014} derive a value for \eps from a probability distribution over a set of negative events and the cost for compensation of affected participants. Our approach avoids the ambiguity of selecting bad events. Abowd and Schmutte~\cite{abowd2019} describe a social choice framework for choosing \eps, which uses the production possibility frontier of the model and the social willingness to accept privacy and accuracy loss. We part from their work by choosing \eps w.r.t.~the advantage of the strong DP adversary. Eibl et al.~\cite{Eibl2018} propose a scheme that allows energy providers and energy consumers to negotiate DP parameters by fixing a tolerable noise scale of the Laplace mechanism. The noise scale is then transformed into the individual posterior belief of the DP adversary per energy consumer. We part from their individual posterior belief analysis and suggest using the local sensitivity between two datasets that are chosen by the dataset sensitivity heuristic.

The evaluation of DP in a deep learning setting has largely focused on MI attacks~\cite{Bernau2019, RRL+18, Hayes2019, JE19, JWK+20, chen2020, shokri2017}. From Yeom et al.~\cite{YGF+18} we take the idea of bounding membership advantage in terms of DP privacy parameter \eps. However, while MI attacks evaluate the DP privacy parameters in practice, DP is defined to offer protection from far stronger adversaries, as Jayaraman et al.~\cite{JE19} empirically validated. 
Humphries et al.~\cite{Humphries} derive a bound for membership advantage that is tighter than the bound derived by Yeom et al.~\cite{YGF+18} by analyzing an adversary with additional information. Furthermore, they analyze the impact of giving up the i.i.d.~assumption. Their work does not suggest an implementation of the strong DP adversary, whereas our work suggests an implemented adversary.

Jagielski et al.~\cite{JUO20} estimate empirical privacy guarantees based on Monte Carlo approximations. While they use active poisoning attacks to construct datasets $\cali{D}$ and $\cali{D'}$ that result in maximally different gradients under gradient clipping, we define dataset sensitivity, which does not require the introduction of malicious samples. 

\section{Conclusion}
\label{sec:conc}

We defined two identifiability bounds for the DP adversary in ML with DPSGD: maximum posterior belief $\rb$ and expected membership advantage $\ra$. These bounds can be transformed to privacy parameter \eps. In consequence, with $\ra$ and $\rb$, data owners and data scientists can map legal and societal expectations w.r.t.~identifiability to corresponding DP privacy parameters. Furthermore, we implemented an instance of the DP adversary for ML with DPSGD and showed that it allows us to audit parameter \eps. We evaluated the effect of sensitivity in DPSGD and showed that our upper bounds are reached under multidimensional queries with composition. To reach the bounds, sensitivity must reflect the local sensitivity of the dataset. We approximate the local sensitivity for DPSGD with a heuristic, improving the utility of the differentially private model when compared to the use of global sensitivity.  
\section{Acknowledgements}
\label{sec:ack}
This work has received funding from the European Union’s Horizon 2020 research and innovation program under grant agreement No.~825333 (MOSAICROWN), and used datasets from the UCI machine learning repository~\cite{DG19}.

\bibliographystyle{abbrv}
\bibliography{DPDL}

\end{document}